\documentclass[a4paper,10pt]{article}
\usepackage{amssymb}
\usepackage{amsfonts}
\usepackage{amsmath}
\usepackage{amsthm}
\usepackage[sort,numbers]{natbib}
\usepackage{color, graphicx}

\newcommand{\BR}{\mathbb{R}}

\newcommand{\mX}{{\cal X}}

\newcommand{\mB}{{\cal B}}

\def\boxit#1{\vbox{\hrule\hbox{\vrule\kern6pt
          \vbox{\kern6pt#1\kern6pt}\kern6pt\vrule}\hrule}}
\newtheorem{theorem}{Theorem}[section]
\newtheorem{lemma}{Lemma}[section]

\newtheorem{corollary}{Corollary}[section]

\begin{document}

\thispagestyle{empty}
\title{Bayesian Fusion Estimation via t-Shrinkage}
\author{Qifan Song and Guang Cheng
\thanks{Qifan Song is an assistant professor in Department of Statistics, Purdue University. Email: qfsong@purdue.edu;
Guang Cheng is a professor in Department of Statistics, Purdue University. Email: chengg@purdue.edu.}
}

\date{}

\maketitle

\begin{abstract}
Shrinkage prior has gained great successes in many data analysis, however, its applications mostly focus on the Bayesian modeling of  sparse parameters. In this work, we will apply Bayesian shrinkage to model high dimensional parameter that possesses an unknown blocking structure. We propose to impose heavy-tail shrinkage prior, e.g., $t$ prior, on the differences of successive parameter entries, and such a fusion prior will shrink successive differences towards zero and hence induce posterior blocking. Comparing to conventional Bayesian fused lasso which implements Laplace fusion prior, $t$ fusion prior induces stronger 
shrinkage effect and enjoys a nice posterior consistency property.
Simulation studies and real data analyses show that $t$ fusion has superior performance to the frequentist fusion estimator and Bayesian Laplace-fusion prior.
This $t$-fusion strategy is further developed to conduct a Bayesian clustering analysis, and simulation shows that the proposed algorithm obtains better posterior distributional convergence 
than the classical Dirichlet process modeling.
\end{abstract}


\section{Introduction}\label{intro}
High dimensionality plays an important role in modern statistical applications such as genomics, image processing, finance and etc. An overview for the development of high dimensional analysis can be found in \cite{van2011statistics} and references therein. To overcome ill-posed problems that involve high dimensional parameters, one usually assumes that the true 
parameter value lies in a low dimensional subspace. To obtain such low dimensional estimation, the idea of regularization is commonly used, via penalized likelihood approaches or using informative prior specifications. Various penalty functions have been proposed for consistent frequentist estimation, including Lasso \cite{Tibshirani1996}, SCAD \cite{FanL2001}, adaptive Lasso \cite{Zou2006} and MCP \cite{Zhang2010}. For high dimensional Bayesian inferences, sparsity induced prior, such as spike-and-slab prior \cite{Jiang2007,LiangSY2013, NarisettyH2014,SongL2014, YangWJ2015,CastilloSHV2015,ScottB2010,JohnsonR2012}, is widely
used for model selection. 
Fused Lasso \cite{tibshirani2005sparsity} considers another type of low dimensional embedding of a high dimensional parameter $\theta=(\theta_i)_{i=1}^p$ where the successive differences $\vartheta_i = \theta_{i}-\theta_{i-1}$ are assumed to be sparse as well, in other words, there exists a consecutive block partition of $\theta_i$'s, such that $\theta_i$'s are constant within each block.
Fused lasso method proposes a penalty function $\lambda_1\sum|\theta_i|+\lambda_2\sum |\vartheta_i|$
which consists of two terms that encourage sparsity among $\theta_i$'s and $\vartheta_i$'s respectively.

In this work, we consider the following Gaussian mean problem:
\begin{equation}\label{model}
 y_i=\theta_i^*+\varepsilon_i
\end{equation}
where $\varepsilon_i$'s are iid normal error with unknown variance $\sigma^2$. Similarly to Fused Lasso applications, we also assume true parameter $\theta^*$ is blocky in the sense that 
there exists a partition $\{\mB_1^*,\dots, \mB_{s}^*\}$
of $\{1,...,n\}$ such that $\theta^*_i$'s are constant
for all $i\in\mB_k^*$. Correspondingly, we define set $G^*=\{2\leq i\leq n: \vartheta_i^*:=\theta_i^*-\theta_{i-1}^*\neq 0 \}$ whose number of elements is supposed to be much smaller than $n$.
We are interested in conducting Bayesian structure recovery of $\theta^*$.
Motivated by the $L_1$-fusion penalty used by Fused lasso estimator \cite{tibshirani2005sparsity}, as well as the development of the Bayesian lasso \cite{ParkC2008}, \cite{kyung2010penalized} introduced Bayesian fused lasso by imposing independent Laplace priors on all successive differences. 
The implementation of Laplace shrinkage prior can significantly reduce the posterior sampling costs comparing to spike-and-slab modeling, and conceptually, the Laplace prior can be nicely interpreted as a Bayesian counterpart of $L_1$ penalty.
However, many recent Bayesian theoretical developments show that
in the context of sparse linear regression models,
Laplace prior fails to
achieve satisfactory posterior contraction \cite{CastilloSHV2015, BhattacharyaPPD2015, SongL2017}.
It is believed that the posterior inconsistency of Laplace prior is due to its exponentially light tail, and
\cite{SongL2017} suggests to use heavy tail prior distribution for sparse linear regression models, which can induce sufficient Bayesian shrinkage effect and thereafter guarantee to recover the sparsity structure.

We find that the above phenomenon holds for Bayesian fusion estimation as well: imposing Laplace prior on $(\theta_i-\theta_{i-1})$ leads to a smoothly varying $\theta_i$ estimation rather than a blocky $\theta$, thus it fails to identify the blocking structure.
Therefore, in this paper, we propose to use independent student-$t$ priors on successive differences $\theta_i-\theta_{i-1}$ for a Bayesian fusion problem. Our results show that such a simple $t$ fusion Bayesian modeling leads to very accurate posterior estimation. More importantly, comparing with Laplace prior or frequentist $L_1$ penalization,
its performance on detecting the blocking structure is much better. The asymptotic posterior convergence induced by $t$ fusion prior is investigated as well. A related Bayesian work is \cite{shimamura2018bayesian} who proposed to use a Normal-Exponential-Gamma (NEG) prior for the successive differences. However, their Bayesian inference is only based on the maximum a posterior (MAP) estimator, while our application tries to fully utilize the whole posterior distributional information.

Furthermore, we consider a practically useful extension to Bayesian fusion estimation. Instead of assuming that $\theta^*$ has a {\em consecutive} blocking structure, it is more realistic to assume that $\theta^*$ possesses an unknown clustering structure. In other words, some unnecessarily consecutive $\theta_i^*$'s share the same value within an $\mB_s^*$.
In a broader scope, such a clustering problem can viewed as a simplest example of 
subgroup analysis where we assume that a subject-related parameter $\theta$ follows an unknown
grouping structure. For example, in clinical trial studies, the treatment effects may vary across different subpopulations,
but remain the same for the patients belonging to the same subpopulation.
If one can correctly identify the subpopulation structure, then specific medical therapies can be prescribed for each subpopulation to maximize 
the treatment effectiveness. 
The existing Bayesian clustering analysis \cite{wade2015bayesian,heller2005bayesian,Mozeika2018,berger2014bayesian} usually impose discrete priors on the clustering structure, along with a conditional prior on $\theta$ given specific clustering structure. In contrast, we propose to directly model the parameter $\theta$ via $t$ fusion prior.
Such a prior specification allows a computationally efficient Gibbs posterior sampling algorithm. Our simulation shows that our procedure yields reasonable cluster structure recovery, and moreover it beats the usual Dirichlet process prior in terms of posterior contraction.

This paper is organized as follows.
In Section \ref{secmain}, we study the Bayesian fusion problem with $t$ prior specification. We will present the posterior asymptotic result, and discuss its difference from the Laplace prior.
In Section \ref{secada}, we will use the $t$-fusion prior to solve the clustering problem. 
Several simulation studies and one real data application are presented 
in Section \ref{simu}.
At last, Section \ref{end} provides more discussions and remarks.
All technical proofs are postponed to the appendix section.

Throughout this work, the following notation is used.
Given two positive sequences $\{a_n\}$ and $\{b_n\}$, 
$a_n\succ b_n$ means $\lim(a_n/b_n)=\infty$ and $a_n\asymp b_n$ means $-\infty<\lim\inf(a_n/b_n)\leq\lim\sup(a_n/b_n)<\infty$.
$\|x\|$ and $\|x\|_1$ denote $L_2$ and $L_1$ norms of vector $x$.

\section{Bayesian fusion via $t$-shrinkage}\label{secmain}
\subsection{Bayesian Modeling}\label{secmodel}
Suppose we observe independent data $\{y_i: i=1,\dots,n\}$ following model (\ref{model}). The indexing of the data has certain practical or scientific meaning, under which we can assume that the parameter vector
$\theta^*$ is ``stepwise", in the sense that most of the successive differences
$\vartheta_i=\theta_i-\theta_{i-1}$ are exactly 0.
To induce the sparsity for both $\theta_i$'s and $\vartheta_i$'s, \cite{tibshirani2005sparsity} proposed the following fused lasso estimator
\[
\widehat\theta^{\rm FL}=\arg\min\left(\frac{\|y-\theta\|^2}{2}+\lambda_1\sum_{i=1}^n|\theta_i|+\lambda_2\sum_{i=2}^{n}|\vartheta_i|\right).
\]
If one is not interested in pursuing the sparsity of $\theta$'s, then a fusion estimator \cite{rinaldo2009properties} can be used
\begin{equation}\label{fusest}
\widehat\theta^{\rm F}=\arg\min\left(\frac{\|y-\theta\|^2}{2}+\lambda\sum_{i=2}^{n}|\vartheta_i|\right)=\arg\min\left(\frac{\|y-\theta\|^2}{2}+\lambda\sum_{i=2}^{n}|\theta_{i}-\theta_{i-1}|\right),
\end{equation}
for some tuning parameter $\lambda$. The above objective functions are both convex and fast computation
algorithms are developed, e.g. \cite{liu2010efficient, tibshirani2005sparsity}. The penalty term $\lambda\sum_{i=2}^{n}|\theta_{i}-\theta_{i-1}|$ is can be interpreted as 
the negative logarithm of prior density used for Bayesian inferences, therefore, a nature Bayesian expansion to (\ref{fusest}) is Laplace (double exponential) prior modeling \cite{kyung2010penalized,shimamura2018bayesian}. To account for the unknown variance parameter $\sigma^2$ and $\theta_1$, a convenient prior specification could be 
\begin{equation}\label{Lprior}
\begin{split}
\sigma^2&\sim \mbox{Inverse-Gamma}(a_\sigma,b_\sigma),\quad 
\theta_1|\sigma^2\sim N(0, \sigma^2\lambda_1),\\
(\theta_{i}-\theta_{i-1})|\sigma^2 &\sim \mbox{Laplace}(\lambda/\sigma), \mbox{ for all }i=2,\dots, n
\end{split}
\end{equation}
where Laplace($a$) denotes the distribution with cdf $f(x)\propto \exp(-a|x|)$.

According to \cite{andrews1974scale}, the above Laplace prior $\mbox{Laplace}(\lambda/\sigma)$ can be rewritten as a scale mixture of normal distributions:
\[
(\theta_{i}-\theta_{i-1})|\sigma^2,\lambda_i \sim N(0, \lambda_i\sigma^2),\quad
\lambda_i\sim \mbox{exp}(-\lambda^2/2),
\]
where \mbox{exp}$(a)$ denotes the exponential distribution $f(x)\propto \exp(-ax)$.
This hierarchical representation for Laplace prior leads to a Gibbs sampling 
update that is similar to the Bayesian lasso \cite{ParkC2008}:
\begin{equation}\label{gibbs1}
\begin{split}
\lambda_i^{-1}\sim  &\mbox{Inverse-Gaussian}(\lambda\sigma/|\theta_i-\theta_{i-1}|, \lambda^2) \mbox{ for all }i=2,\dots, n,\\
\sigma^2\sim&\mbox{Inverse-Gamma}
\left(a_\sigma+n,b_\sigma+\frac{\|y-\theta\|^2}{2}+\frac{\theta_1^2}{2\lambda_1}
+\sum_{i=2}^{n}\frac{(\theta_i-\theta_{i-1})^2}{2\lambda_i}\right),\\
\theta_i\sim&N( \mu_i,\nu_i),
\end{split}
\end{equation}
where 
$\nu_i^{-1}=1/\sigma^2+1/\lambda_{i+1}\sigma^2+1/\lambda_i\sigma^2$ and 
$\mu_i = \nu_i(y_i/\sigma^2+\theta_{i+1}/\lambda_{i+1}\sigma^2+\theta_{i-1}/\lambda_i\sigma^2)$, $\lambda_{n+1}$ and $\lambda_{0}$ are considered to be infinite, and 
$\mbox{Inverse-Gaussian}(a,b)$ denotes inverse Gaussian distribution with cdf
$f(x)\propto x^{-3/2}\exp[-b(x-a)^2/(2a^2x)]$.

Despite the popularity of Laplace prior in many applications, recent Bayesian works \cite{CastilloSHV2015, BhattacharyaPPD2015, SongL2017} point 
out that, if we impose independent
$\theta_i|\sigma^2\sim \mbox{Laplace}(\lambda/\sigma)$ in high dimensional sparse regression under Laplace priors, then the induced posterior has only a sub-optimal contraction rate, or even diverges.
In other words, the posterior distribution of 
$\theta$ doesn't contract into a small neighborhood around true value $\theta^*$ appropriately. For a blocky parameter $\theta$, we observe similar empirical results,
as showed in the toy example in Section \ref{sectoy}: the Laplace fusion prior fails to
shrink the observations, which belongs to the same block, towards a same value. Hence, the resultant Bayesian estimate of $\theta$ doesn't have a step-wise pattern at all.

Following the theoretical discovery of \cite{SongL2017}, we consider using a class of heavy tailed 
priors for the successive differences $\vartheta_i$'s. Specifically, this work will
assign an $t$-shrinkage prior:
\begin{equation}\label{tprior}
\begin{split}
\sigma^2&\sim \mbox{Inverse-Gamma}(a_\sigma,b_\sigma),\quad 
\theta_1|\sigma^2\sim N(0, \sigma^2\lambda_1),\\
(\theta_{i}-\theta_{i-1})|\sigma^2 &\sim t_{df}(s\sigma), \mbox{for all }i=2,\dots, n
\end{split}
\end{equation}
where $t_{a}(b)$ denotes $t$-distribution with degree of freedom $a$ and scale parameter $b$. Note that the above $t$ distribution can be rewritten as an inverse-gamma scaled Gaussian mixture as
\[
(\theta_{i}-\theta_{i-1})|\sigma^2,\lambda_i \sim N(0, \lambda_i\sigma^2),\quad
\lambda_i\sim \mbox{Inverse-Gamma}(a_t, b_t),
\]
where $a_t$, $b_t$ satisfy $df=2a_t$ and $s=\sqrt{b_t/a_t}$.
Under this $t$-prior, the posterior distribution still allows 
a full conditional Gibbs sampler, where the update for $\theta_i$'s
and $\sigma^2$ are exactly the same as in (\ref{gibbs1}) and the
update of $\lambda_i$'s follows
\[
\lambda_i\sim  \mbox{Inverse-Gamma}\left(a_t+1/2,b_t+\frac{(\theta_i-\theta_{i-1})^2}{2\sigma^2}\right)
\mbox{ for all }i=2,\dots, n.
\]

To further understand the difference between the Laplace fusion prior and 
$t$ fusion prior, we compare their conditional prior 
$\pi(\theta_i|\theta_{i-1},\theta_{i+1}, \sigma)$. Figure \ref{priorshape}
plots the function $-\log[\pi(\theta_i|\theta_{i-1}=-1,\theta_{i+1}=1, \sigma=1)]$, up to a constant, for both prior specifications.
It is clear to see that the conditional $t$-fusion prior allocates most
of its prior mass at the two small neighborhoods centered at 
$\theta_{i-1}$ and $\theta_{i+1}$, given a sufficiently small scale parameter
$s$. In other word, the prior introduces a strong shrinkage effect on $\theta_i$, towards either $\theta_{i-1}$ or $\theta_{i+1}$. 
Therefore, for all $i=2,\dots,n$, $\theta_i$ will merge with either $\theta_{i-1}$ or $\theta_{i+1}$ in the posterior distribution, which thereafter induces a posterior blocking structure.
On the other hand, the conditional Laplace-fusion prior has a uniform
prior density within the interval $[\theta_{i-1},\theta_{i+1}]$.
Hence, it doesn't encourage the posterior of $\theta_i$ to be grouped with either $\theta_{i-1}$ or $\theta_{i+1}$.
It is worth to mention that the NEG-fusion prior \cite{shimamura2018bayesian}
also has a similar plot pattern for its conditional prior density function. 
But a critical difference between $t$ prior and NEG prior, is that the density for NEG prior is non-differentiable at 0. Therefore, the MAP of NEG-fusion prior possesses an exact blocking structure, and \cite{shimamura2018bayesian} only use this MAP for Bayesian inferences rather than the whole posterior distributional information.
But the $t$ prior is continuous differentiable everywhere (the functions displayed in the upper plot of Figure \ref{priorshape} is actually smooth at -1 and 1). Hence, its MAP doesn't have blocky structure, and in this work we will utilize all the posterior samples for the Bayesian analysis.

\begin{figure}[htb]
	\begin{center}
		\includegraphics[width=8cm]{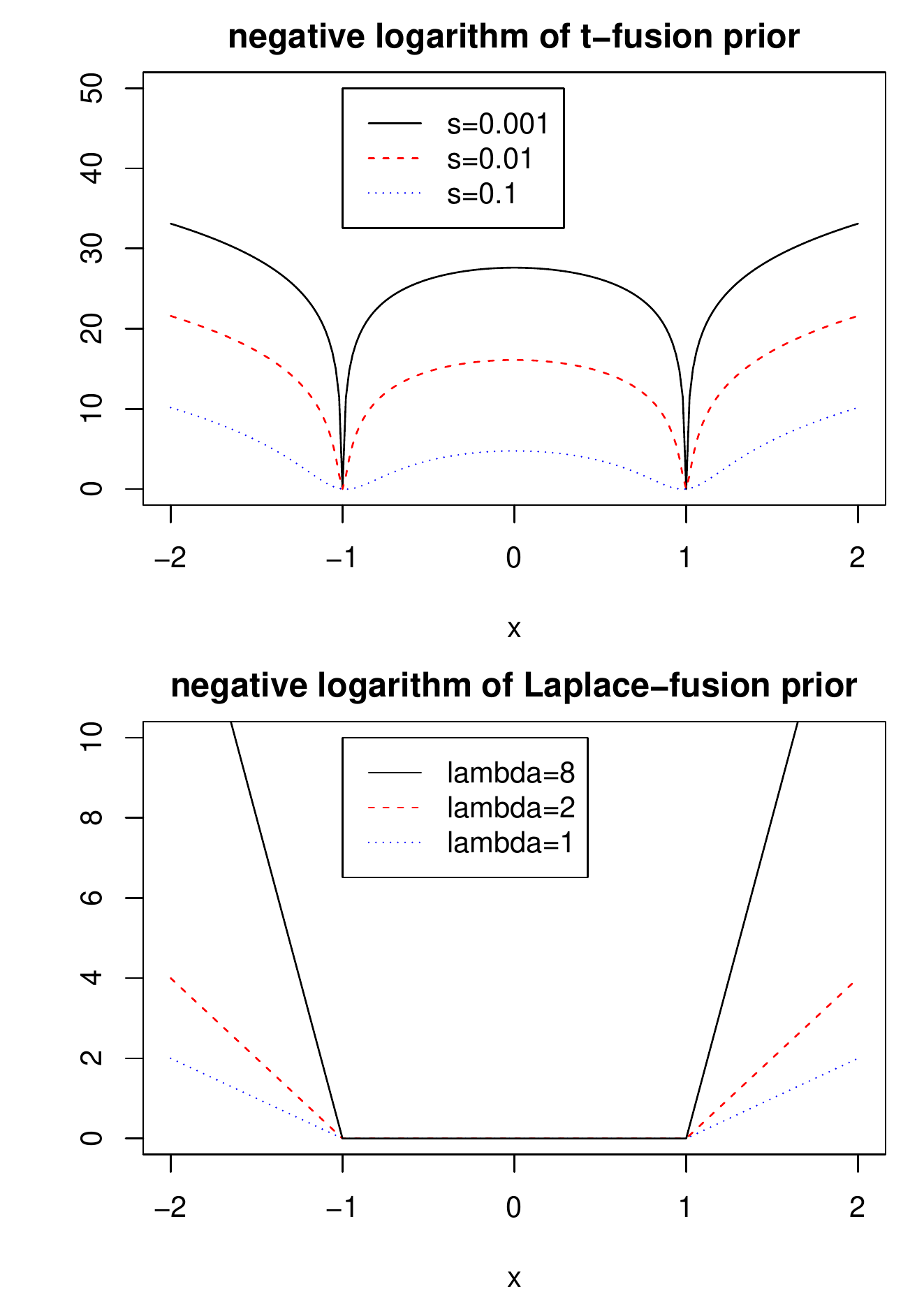}\vskip -0.3in
		\caption{
		The negative logarithm of conditional prior  $-\log[\pi(\theta_i|\theta_{i-1}=-1,\theta_{i+1}=1, \sigma=1)]$ under
		difference hyperparameter values.
		}\label{priorshape}
	\end{center}
\end{figure}

\subsection{Posterior contraction of Bayesian $t$-fusion}
In this section, we study the theoretical performance induced by the $t$-fusion prior specification (\ref{tprior}). Our theoretical investigation follows the framework of \cite{SongL2017}, which studies the posterior convergence rate of coefficient $\beta$ in high dimensional sparse regression models $y=X\beta+\varepsilon$.
Note the model (\ref{model}) can also be represented as a sparse linear regression, where the design matrix
$X$ is a $n$ by $n$ matrix whose lower triangle entries are all 1 and $\beta=(\theta_1,\vartheta_2,\dots,\vartheta_n)$ is a unknown sparse vector.
The following theorem studies the general posterior convergence properties given an independent prior over all $\vartheta_i$'s.
\begin{theorem}[Posterior consistency]\label{thm1}
Assume that $|G^*|\prec n/\log(n)$, and prior specification follows that 
$\sigma^2\sim \mbox{Inverse-Gamma}(a_\sigma,b_\sigma)$, $\theta_1$ and $\vartheta_i$'s are conditionally independent given $\sigma^2$ with prior density $\pi(\theta_1,\vartheta_i's|\sigma)\propto (1/\sigma)^nf_\theta(\theta_1/\sigma)\prod_{i=2}^n f_\vartheta(\vartheta_i/\sigma)$.
Furthermore, if
\begin{equation}\label{cond1}
    \begin{split}
    &\int_{-|G^*|\log(n)/n^2}^{|G^*|\log(n)/n^2}f_\vartheta(x)dx\geq 1- n^{1+u},\mbox{  for some $u>0$, }\\
    &-\log(\underline\pi_\vartheta)=O(\log n),\quad \mbox{ where } 
         \underline\pi_\vartheta = \min_{|x|\leq \max_i|\vartheta^*_i/\sigma^*|+1} f_\vartheta(x),\\
    &-\log(\underline\pi_\theta)=O(|G^*|\log n),\quad \mbox{ where } 
         \underline\pi_\theta = \min_{|x|\leq |\theta^*_1/\sigma^*|+1} f_\theta(x),\\
    &a_\sigma\log(1/b_\sigma)+b_\sigma/\sigma^{*2}+(a_\sigma+2)\log(\sigma^{*2})=O(|G^*|\log n),
    \end{split}
\end{equation}
then there exist a constant $M$, and $\epsilon_n\asymp \sqrt{|G^*|\log n/n}$, such that the posterior distribution satisfies
\[
\pi(\|\theta-\theta^*\|\geq M\sigma^*\sqrt{n}\epsilon_n|y ) \rightarrow 0,
\]
where the convergence holds in probability or in $L_1$ w.r.t. the probability measure of $y$.
\end{theorem}
The proof closely follows the Theorem A.1 in \cite{SongL2017}, and for the sake of readability, the proof is provided in the Appendix.
The first inequality of sufficient condition set (\ref{cond1}) requires that
the prior imposed on $\vartheta_i$'s is highly concentrated around zero, such that it induces sufficient Bayesian shrinkage effect for those $\vartheta_i$'s whose true value is 0. The rest inequalities of (\ref{cond1}) essentially require that
the prior density at true parameter is at least of order $e^{-cn\epsilon_n^2}$ for some $c$, and this helps to prevent over-shrinkage for those $\vartheta_i$'s whose true value is not 0.
Similar conditions, which need the prior to be ``thick" at true parameter values, 
are regularly used in Bayesian literature \citep{Jiang2007,KleijnV2006,GhosalGV2000,GhosalV2007}. Given the concrete forms for prior density $f_\theta$ and $f_\vartheta$, the second and third inequalities of (\ref{cond1}) are equivalent to some upper bound constraints on the magnitude of
$\theta_1^*$ and $\max_i|\vartheta_i^*|$'s (see e.g., Corollary \ref{thmt}). The fourth inequality of (\ref{cond1}) trivially holds for any fixed $a_\sigma$ and $b_\sigma$ if $\sigma^{*2}$ is assumed to be a constant. If the unknown error variance is supposed to be varying with respect to $n$, e.g.,
the studies of Gaussian sequence models \cite{johnstone2010high} commonly assume that $\sigma^{*2}\propto n^{-1}$, then one can choose a fixed $a_\sigma$ and 
$b_\sigma\asymp n^{-\kappa}$ for some $\kappa>0$. Under such a choice, condition (\ref{cond1}) holds as long as $-k\log n\leq \log(\sigma^{*2})<K\log n$ for some positive constant $K$.

The above result states that almost all the posterior mass contract into
a neighborhood of $\theta^*$ with a radius $M\sigma^*\sqrt{n}\epsilon_n$, that 
is, the posterior convergence rate is of an order $\sigma^*\sqrt{|G^*|\log n}$. Note that if
the partition index set $G^*$ were known, the oracle rate of contraction turns out to be $O(\sigma^*\sqrt{|G^*|})$. Hence, the Bayesian shrinkage achieves the ideal risk up to a logarithmic term in $n$. In frequentist literature, Theorem 2.7 of \cite{rinaldo2009properties} showed that the convergence rate of Fused Lasso is no larger than $O(\sigma^*\sqrt{|G^*|\log |G^*|})$. However, this rate is not directly comparable since an additional minimal signal strength
condition, which ensures that $G^*$ can be fully recovered in probability, is imposed in this paper. 

The posterior distribution of $\vartheta$ are always continuous, and doesn't directly provide Bayesian inferences for the block structure, or equivalently, the unknown $G^*$. The following result characterizes the some asymptotic performance of posterior block partition via discretization.
\begin{theorem}[Posterior selection]\label{thm2}
Assume the conditions of Theorem \ref{thm1} hold, denote
$G(\theta,\sigma)=\{i:|\vartheta_i/\sigma|<\epsilon_n/n\}$, then the posterior of $G(\theta,\sigma)$
satisfies
\[
\pi(\{|G(\theta,\sigma)\backslash G^*|>\delta|G^*|  \} |y)\rightarrow 0,
\]
for some fixed constant $\delta$, where the convergence holds in probability and in $L_1$.
\end{theorem}

Therefore, the posterior distribution of $G(\theta,\sigma)$, which is induced by the posterior of ($\theta, \sigma^2)$ and mapping $G(\cdot,\cdot)$, can be treated as a discrete posterior distribution for the unknown $G^*$ and used for Bayesian block partition selection. Theorem \ref{thm2} essentially claims that the number false positive selections of via discretization $G(\cdot,\cdot)$ is bounded in posterior probability. Such a result is comparable to the model selection behavior of Bayesian Dirichlet-Laplace
shrinkage \cite{BhattacharyaPPD2015}. It is worth to note that if certain minimal signal strength condition holds as well, i.e., $\min_{\vartheta_i^*>0}|\vartheta_i^*|$ is bounded away from zero, with addition assumptions, one can derive an even stronger posterior selection consistency result that $\pi(G(\theta,\sigma)=G^*   |y)\rightarrow^p 1$, following the proof of Theorem 2.3 in \cite{SongL2017}. Readers of interests can easily derive such posterior selection consistency by themselves.

It is not difficult to verify the imposed conditions for the proposed $t$ fusion shrinkage prior (\ref{tprior}).
\begin{corollary}\label{thmt}
When $f_\theta$ is the cdf of normal distribution $N(0,\lambda_1)$ and $f_\vartheta$ is the cdf of a $t$-distribution with scale parameter $s$, then the first three inequalities of
(\ref{cond1}) hold when
\[
\begin{split}
&\log(\max_i|\vartheta_i^*|/\sigma^*)=O(\log(n)),\\
&\theta_1^{*2}/(\lambda_1\sigma^{*2})+\log(\lambda_1)=O(\log n),
\end{split}
\]
and $s=n^{-c}$ for some sufficiently large $c$.
\end{corollary}

The above results hold not only for the scaled $t$ prior, but also for any other choice of $f_\vartheta$, as long as $f_\vartheta$ has a polynomially decaying tail. Note that the above results can not be generalized to light tailed distributions such as Laplace-fusion prior,
since it will lead to an unrealistic sufficient condition that 
$\max_i|\vartheta_i^*|/\sigma^*=o(1)$. 

\subsection{Bayesian posterior inference}\label{sectoy}
In this section, we will illustrate the posterior behavior of Bayesian $t$-fusion by a 
toy example, and compare it to Bayesian Laplace-fusion. We will also discuss other issues related to hyperparameter choice.

A simulation data was generated with $n=100$ and $\sigma^*=0.5$. The data and underlying true parameter value are plotted in Figure \ref{toy1}.
Three estimation procedures are considered: 1) $L_1$ fusion estimation (\ref{fusest}) where 
the tuning parameter are selected by cross validation; 2) Bayesian Laplace prior (\ref{Lprior}) with $\lambda = \sqrt{2\log (n)} $ and $\lambda_1$=5; 3) Bayesian t prior (\ref{tprior})
with $a_t=2$, $b_t=0.005$, $a_\sigma=0.5$, $b_\sigma=0.5$ and $\lambda_1=5$.
The posterior samples are obtained by Gibbs sampler for 2000 iterations.
The frequentist estimator and boxplots of Bayesian marginal posterior distributions are 
displayed in Figure \ref{toy1}.
Note that for the sake of readability of the figure, we don't draw the outliers in these boxplots.

\begin{figure}[htb]
	\begin{center}
		\includegraphics[width=12cm]{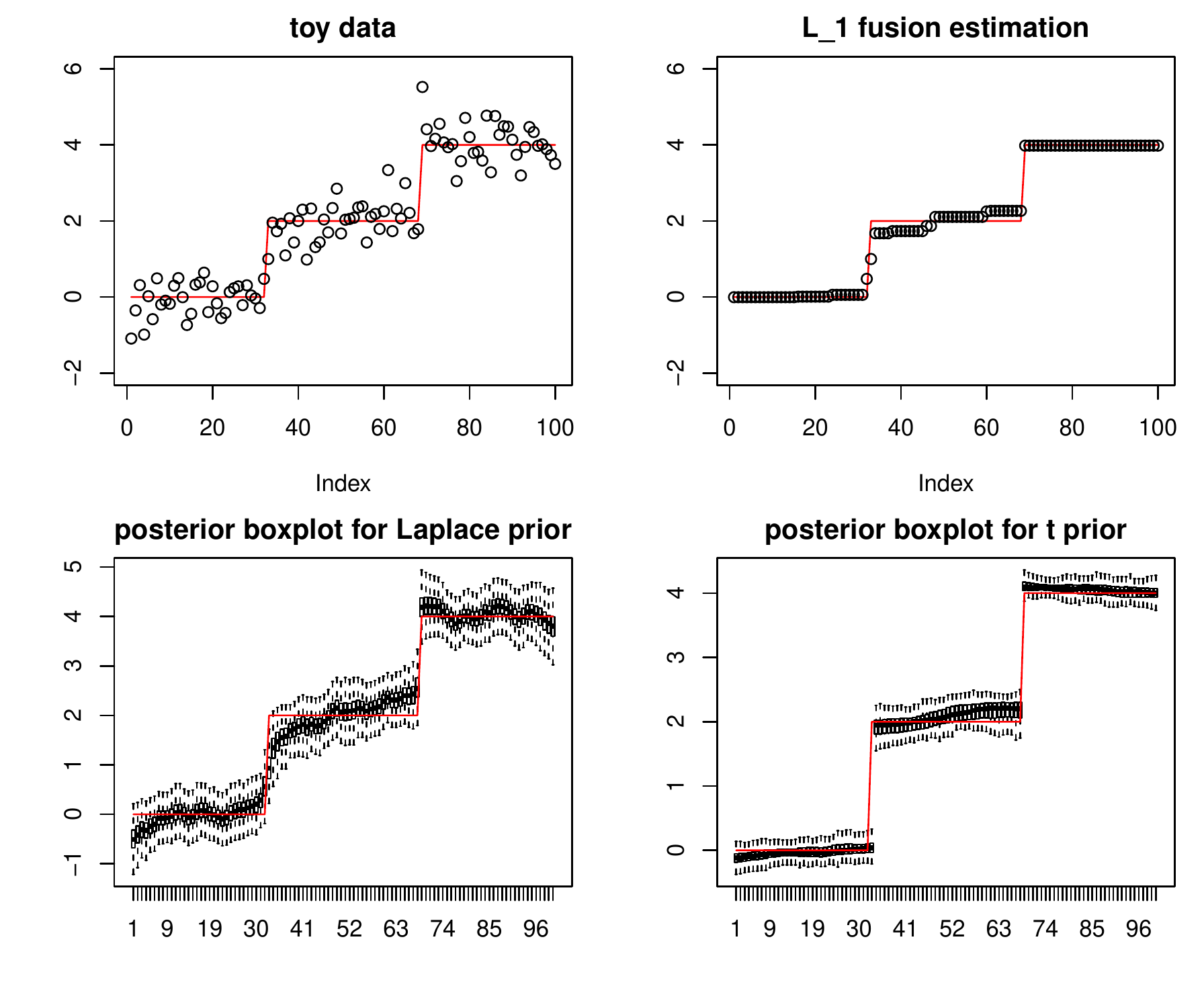}\vskip -0.3in
		\caption{
			Upper Left: Simulated toy data; Upper right: Frequentist fusion estimation 
			(\ref{fusest}); Lower left: Marginal posterior boxplots for each $\theta_i$
			under Laplace fusion prior; Lower right: Marginal posterior boxplots under $t$ fusion prior. The red curve denotes the true $\theta^*_i$'s which contain three blocks.
		}\label{toy1}
	\end{center}
\end{figure}

The comparison shows that the $L_1$ fusion penalty leads to a strictly sparse $\widehat\vartheta$ estimation and the estimated $\widehat\theta$ does have a blocking structure, but there is a mild over-partition issue.
Due to the nature of Bayesian shrinkage prior, the two Bayesian approaches only produce continuous posterior samples of $\vartheta_i$.
Comparing with the Laplace prior, the advantages of $t$ prior are quite 
obvious. Its posterior concentration is better, i.e., shorter boxes and whiskers for the boxplots, and posterior mean is much more close to the true red curve. Although its posterior estimation is not exactly blocky, one can
visually identify the three blocks. 
In contrast, the Laplace prior generates larger posterior variance, and its posterior center smoothly fluctuates around the true step function.
As discussed at the end of Section \ref{secmodel}, Laplace fusion prior itself merely has shrinkage effect for the successive differences, and hence the data fluctuation carries over to the posterior distribution of $\theta$. In this case, it is not clear to recover the underlying block structure.

Note that for both (\ref{Lprior}) and (\ref{tprior}) prior specifications, their posterior distributions highly depend on the choice of hyperparameter, especially the scale parameters of the Laplace (parameter $1/\lambda$) and $t$ prior (parameter $s$) distribution. To understand the influence of 
the scale parameter, we increase and decrease the scale parameter by a factor of $\sqrt{10}$, for both Laplace and $t$ priors. The corresponding posteriors are displayed
in Figure \ref{toy2}.

\begin{figure}[htb]
	\begin{center}
		\includegraphics[width=12cm]{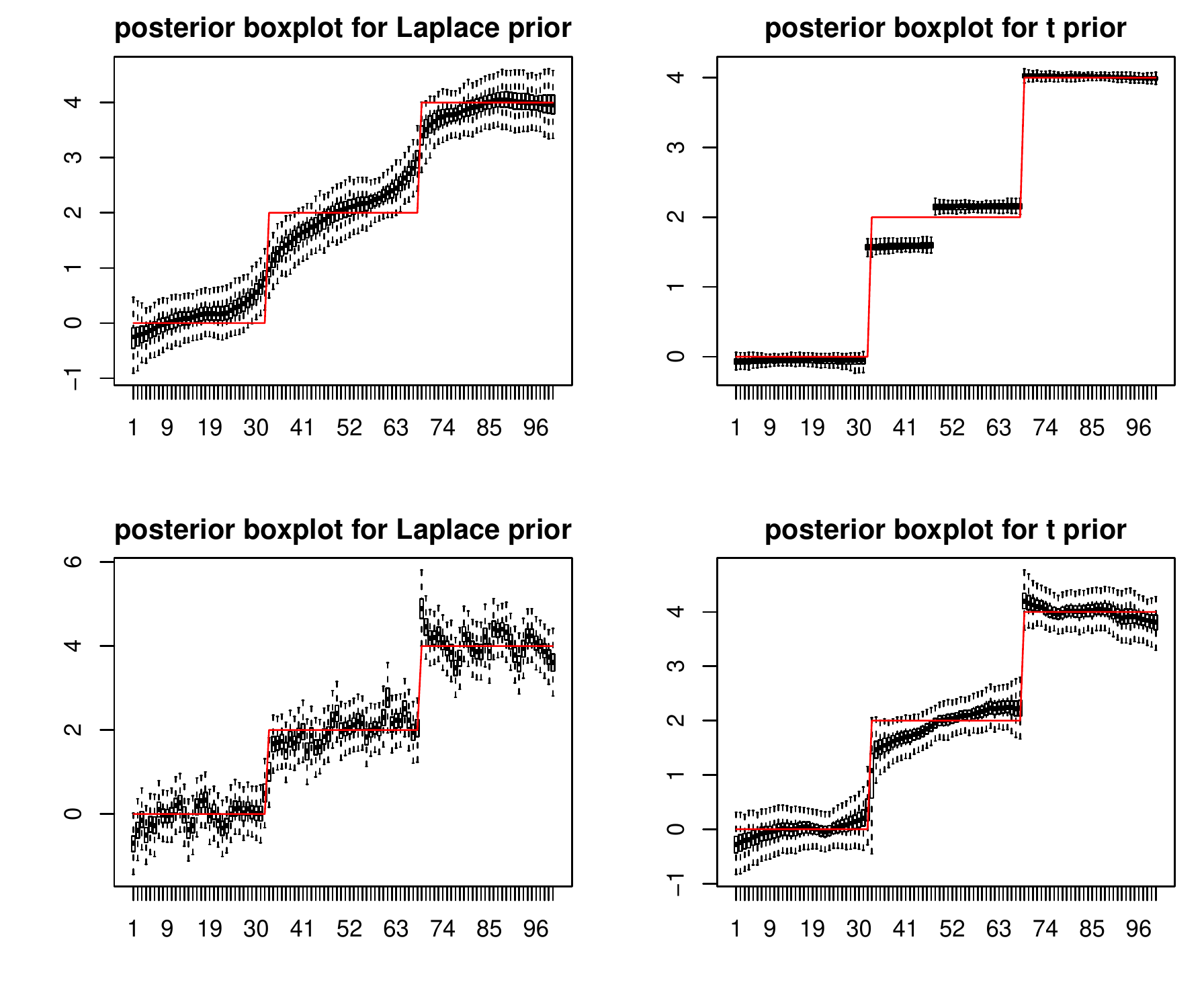}\vskip -0.3in
		\caption{
			Bayesian posterior with different choice of scale parameter.
			Upper Two: The scale parameter decreases by a factor of $\sqrt{10}$.
		    Lower Two: The scale parameter increases by a factor of $\sqrt{10}$.
		}\label{toy2}
	\end{center}
\end{figure}

From Figure \ref{toy2} and its comparison with Figure \ref{toy1}, we see that, for $t$ prior specification, smaller scale hyperparameter leads to smaller posterior variation, and vise versa. 
Choosing an overly large scale parameter weakens the shrinkage effect, thus it fails to shrink the $\theta_i$'s that belong to the same block towards same value, and the corresponding posterior estimation is not flat within blocks. On the other hand, choosing an overly small scale hyperparameter, although yields very strong shrinking and grouping effect, may
potentially over-partition the data.
In conclusion, the scale parameter of $t$ fusion prior controls the aggressiveness of posterior
block partition.
For Laplace prior, both increasing or decreasing scale parameter can not remedy its poor
behavior demonstrated in Figure \ref{toy1}. The choice of scale parameter only affects the smoothness of posterior of $\theta$. Smaller scale parameter leads to smoother non-parametric estimation which is similar to smooth spline regression; larger scale parameter leads to a rather rugged estimation.
This phenomenon can be partially explained by Figure \ref{priorshape}.
As the scale of the Laplace prior decreases (i.e., $\lambda$ increases), the 
conditional prior of $\theta_i$ given $(\theta_{i-1},\theta_{i+1})$ acts more and more like a uniform distribution over $[\theta_{i-1},\theta_{i+1}]$.
Hence, conditional on $(\theta_{i-1},\theta_{i+1})$, the posterior of ``local" total variation, $|\theta_{i+1}-\theta_i|+|\theta_{i}-\theta_{i-1}|$, always
decreases to the minimum value $|\theta_{i+1}-\theta_{i-1}|$ when the scale parameter decreases to zero, regardless of the observation value $y_i$.

As showed in the previous toy example, the hyperparameter value plays an important role
for the performance of Bayesian posterior inferences. 
The theoretical suggestion, i.e., the first inequality of (\ref{cond1}), gives a very small scale parameter, which in practice leads to severe over-partition issue under moderate $n$.
Unlike frequentist high dimensional penalization estimation whose tuning parameter is usually determined by cross validation or selection criterion such as EBIC \cite{ChenC2008, ChenC2012}, 
choosing a proper prior hyperparameter posts many difficulty for Bayesian statisticians.
Conventional choices include imposing a hyper-prior on the hyperparameter 
\cite{VanSV2017,Castillov2012} or empirical Bayes \cite{robbins1985empirical}.
For high dimensional sparse GLMs, \cite{LiangSY2013} suggested choosing a hyperparameter such that the posterior mean and posterior mode are close. Several Bayesian works \cite{shimamura2018bayesian,betancourt2017bayesian}
considered only using the sparse MAP as the Bayesian estimator, and thus abandon the whole posterior distributional information. 
This strategy somehow reduces the Bayesian computation to a 
frequentist optimization problem, and then the hyperparameter can be determined 
by EBIC criterion. It is beyond the scope of this work to theoretically study how to select an appropriate hypereparameter, i.e., the scale parameter of the $t$ prior. An empirical suggestion based on authors' experience is to choose the scale parameter of (\ref{tprior}) such that 
\begin{equation}\label{tune}
P(|t_{df}(s)|\geq \sqrt{\log (n)/n}) \approx 1/n.
\end{equation}
Although it doesn't quite satisfy the conditions in Corollary \ref{thmt}
 which suggests $s=n^{-c}$ for some large $c$, but in practice, it indeed yields 
a reasonable and stable Bayesian performance. Note that our prior choice $a_t=2$, $b_t=0.005$ approximately satisfies (\ref{tune}).

\section{Bayesian $t$-shrinkage for Bayesian clustering}\label{secada}
In this section, we would like to extend the applications of Bayesian $t$ fusion shrinkage 
to Bayesian clustering problem, where the parameter $\theta$ in model
(\ref{model}) is assumed to follow an unknown clustering structure. 
In other words, the observations $y_i$'s are not organized in a sensible order, and
we no longer assume that the true cluster only consists of consecutive indexes.


A full Bayesian clustering analysis, or subgroup identification, usually imposes a prior on the clustering structure, including the number of clusters and how to partition observations into these clusters. For example, one could consider that the cluster structure arises from a tree splitting process, and impose certain prior on the tree nodes
\cite{berger2014bayesian,heller2005bayesian}. The resultant posterior sampling hence usually requires an inefficiently reversible-jump Metropolis-Hastings move to travel across different clustering models
in the sampling space.
Another common approach is to model a mixture distribution via nonparametric priors such as Dirichlet process or Chinese restaurant process \cite{neal2000markov}.
These mentioned approaches enforce explicit {\it clustery} posterior samples
by directly applying discrete prior on subgroup structure. In this section, we will show how to use Bayesian $t$-shrinkage to induce implicit posterior clustering structure.


To implement shrinkage prior, we need to formulate the problem in a proper statistical modeling such that its parameter possesses certain sparsity feature. Conceptually, imposing shrinkage priors on $\theta_i-\theta_{i-1}$ shall not work for this clustering problem since $\{\theta_i-\theta_{i-1}\}_{i=2}^n$ is no longer a sparse vector. However, if some meaningful permutation $r$ of $\theta_i$ is known
such that $\{\theta_{r(i)}-\theta_{r(i-1)}\}_{i=2}^n$ is indeed a sparse vector, for example, $\theta_{r(i)}$ is ascending or descending, then the problem will reduce to Bayesian
shrinkage fusion studied in the previous section.

However, in practice, such $r(\cdot)$ is generally not available. Thus one simple solution would be to substitute $r(\cdot)$ with some estimator $\widehat r()$. 
For model (\ref{model}), a trivial estimator is the rank statistic of the responses $y_i$, i.e., $y_{\widehat r(i)}=y_{(i)}$ where 
$y_{(i)}$ denotes the order statistics of $y_i$'s.
Thus this leads to hybrid Bayesian modeling, depending on a
frequentist pilot estimator $\widehat r$:
\begin{equation}\label{tprior2}
\begin{split}
\sigma^2&\sim \mbox{Inverse-Gamma}(a_\sigma,b_\sigma),\quad 
\theta_{\widehat r(1)}|\sigma^2\sim N(0, \sigma^2\lambda_1),\\
(\theta_{\widehat r(i)}-\theta_{\widehat r(i-1)})|\sigma^2,\lambda_i &\sim N(0, \sigma^2
\lambda_i), \quad \lambda_i\sim\mbox{Inverse-Gamma}(a_t,b_t)\quad\mbox{for all }i=2,\dots, n.
\end{split}
\end{equation}
The corresponding posterior Gibbs sampler follows:
\begin{equation}\label{gibbs2}
\begin{split}
\lambda_i\sim  &\mbox{Inverse-Gamma}\left(a_t+1/2,b_t+\frac{(\theta_{\widehat r(i)}-\theta_{\widehat r(i-1)})^2}{2\sigma^2}\right)
\mbox{ for all }i=2,\dots, n.\\
\sigma^2\sim&\mbox{Inverse-Gamma}
\left(a_\sigma+n,b_\sigma+\frac{\|y-\theta\|^2}{2}+\frac{\theta_{\widehat r(1)}^2}{2\lambda_1}
+\sum_{i=2}^{n}\frac{(\theta_{\widehat r(i)}-\theta_{\widehat r(i-1)})^2}{2\lambda_i}\right),\\
\theta_{\widehat r(i)}\sim&N( \mu_i,\nu_i),
\end{split}
\end{equation}
where $\nu_i^{-1}=1/\sigma^2+1/\lambda_{i+1}\sigma^2+1/\lambda_i\sigma^2$ and 
$\mu_i = \nu_i(y_{\widehat r(i)}/\sigma^2+\theta_{\widehat r(i+1)}/\lambda_{i+1}\sigma^2+\theta_{\widehat r(i-1)}/\lambda_i\sigma^2)$.
The idea of taking advantage of a pilot estimator $\widehat\theta$ and its rank statistic ${\widehat r(i)}$ has also been implemented in the two-stage frequentist approach \cite{ke2015homogeneity}.

It is worth to mention that posterior consistency Theorems \ref{thm1} and \ref{thm2} for Bayesian shrinkage fusion don't apply to the above Bayesian modeling (\ref{tprior2})
even when $\theta_{\widehat r(i)}$ is in a ascending order, because $\varepsilon_{\widehat r(i)}$ are no longer iid error observations. 
Such a data dependent prior (\ref{tprior2}) can also be interpreted 
as a misspecified Bayesian modeling \cite{kleijn2006misspecification}, where the data is the order statistics $y_{(i)}$ from a mixture distribution, and we model $y_{(i)}$'s as independent Gaussian variables whose mean function $E(y_{(i)})$ is a step function.

Although the prior modeling (\ref{tprior2}) is quite natural and intuitive, there are several issues. First, comparing with the iid observations, it is more difficult to to identify the underlying clustering structure of sorted observations. For example, it is much more difficult to visually identify the 3-cluster structure in the toy data showed in Figure \ref{toy3} than the toy data showed in Figure \ref{toy1}. 
Secondly, (\ref{tprior2}) fails to take account of uncertainty of the estimator $\widehat{r}$, and the strict monotonicity of $y_{\widehat r(i)}$ will always
carry over to the posterior of $\theta_{\widehat r(i)}$.
Thirdly, over-clustering will occur. To understand this, let us consider
that all data $y_i$'s have the same mean, i.e.,
there is only one cluster. But the mean function of the sorted data, $E(y_{(i)})$, is actually a strictly increasing function. 
From the perspective of Bayesian misspecification, our posterior of $\theta_{\widehat r(i)}$'s should contract towards the best step function, in the sense of minimizing the Kullback–Leibler divergence \cite{kleijn2006misspecification},
rather than towards the constant function $E(y_i)$.
This causes the over-clustering, and furthermore, the posterior 
 consistency of $t$ shrinkage prior established in Corollary
 \ref{thmt} doesn't hold anymore.
\begin{theorem}\label{pos1}
Assume that all $y_i\sim N(0,1)$ with known variance $\sigma^*=1$. If the prior (\ref{tprior2}) is used (except that there is no necessity to impose a prior on $\sigma^2$) and the scale parameter satisfies $-\log(b_t)\asymp\log n$, then in high probability
\[
\pi(\|\theta\|\leq \sqrt{M\log n}|y)<1/2,
\]
i.e., the rate-$\sqrt{\log n}$ posterior consistency fails.
\end{theorem}

The negative result in Theorem \ref{pos1} motivates us to propose an adaptive pseudo Bayesian shrinkage approach. Instead of using a fixed estimator $\widehat{r}$, we update the ``working" order $r$ of $\theta$ over iterations. To be specific, we modify the above Gibbs sampling iteration as:
\begin{equation}\label{gibbs3}
\begin{split}
&\mbox{update $\theta_{ r(i)}$'s, $\lambda_i$'s and $\sigma^2$ according to (\ref{gibbs2})},\\
& \mbox{update the $r(\cdot)$ to be the rank statistic of the current sample $\theta$}.
\end{split}
\end{equation} 
In other words, we update $ r(\cdot)$ to the rank statistic of the newly obtained $\theta$. 

There are a couple of rationales behind the algorithm (\ref{gibbs3}).
First, the update of $r$ potentially allows algorithm (\ref{gibbs3}) to incorporate the uncertainty of rank statistic into the posterior sampling. Heuristically, if the current $r(\cdot)$ is close to the true ranking, (\ref{gibbs3}) will shuffle the ordering of $\theta$ within each cluster, 
hence $y_{r(i)}$ acts like independent samples yielded by random permutation within cluster, rather than order statistics.
Secondly,  the sampling algorithm (\ref{gibbs3}) can be somehow connected with a full Bayesian modeling:
\begin{equation}\label{tprior3}
\begin{split}
r&\sim \mbox{Uniform distribution over all possible permutations},\\
\sigma^2&\sim \mbox{Inverse-Gamma}(a_\sigma,b_\sigma),\quad 
\theta_{ r(1)}|\sigma^2\sim N(0, \sigma^2\lambda_1),\\
(\theta_{ r(i)}-\theta_{ r(i-1)})|\sigma^2,\lambda_i &\sim N(0, \sigma^2
\lambda_i), \quad \lambda_i\sim\mbox{Inverse-Gamma}(a_t,b_t)\quad\mbox{for all }i=2,\dots, n.
\end{split}
\end{equation}
Under (\ref{tprior3}), the conditional posterior of $r(\cdot)$ follows 
$\pi(r|\theta,\lambda_i's, y)\propto \exp\{-\sum_{i=2}^{n}(\theta_{ r(i)}-\theta_{ r(i-1)})^2/\lambda_i^2\sigma^2\}$, from which sampling is expensive. However, when most of $\lambda_i^2$s are tiny, the distribution of $\pi(r|\theta,\lambda_i's, y)$ will be highly concentrated around its MAP
which is approximately the rank statistic of $\theta$. Therefore, updating $ r$ to the rank statistic of current sample $\theta$, as in (\ref{gibbs3}), can be viewed as 
a convenient sampling of $\pi(r|\theta,\lambda_i's, y)$ under the
uniform hyperprior of $r$. And such a hyperprior does benefit the posterior asymptotic performance, at least it remedies the posterior inconsistency described in Theorem \ref{pos1} when there is only one underlying cluster.

\begin{theorem}\label{thm3}
Assume that all $y_i\sim N(\theta_0,\sigma^2)$, i.e. $\theta^*$ is a vector
of $\theta_0$'s. If prior (\ref{tprior3}) is used with $b_t = n^{-c}$ for some sufficiently large $c$, and
$\theta_0^{2}/(\lambda_1\sigma^{*2})+\log(\lambda_1)=O(\log n)$
then there exist constant $M$ and $\epsilon_n\asymp \sqrt{\log n/n}$, such that
\[
\pi(\|\theta-\theta^*\|\geq M\sigma^*\sqrt{n}\epsilon_n|y)\rightarrow 0,
\]
where the convergence holds in probability and $L_1$.
\end{theorem}
Further investigations will be pursued to assess the posterior contraction induced by (\ref{tprior3})
when the true parameter $\theta^*$ has multiple-cluster structure.
The minimal cluster difference $\min_{\{(i,j):\theta_i\neq\theta_j\}}|\theta_i-\theta_j|$
shall play an crucial role for the Bayesian asymptotics. 
If the minimal cluster difference is bounded by constant, then even
the best Bayes classifier makes as many as $O(n)$ misclassifications, 
which leads to $\sqrt n$-rate $L_2$ error. And we conjecture that when
the minimal cluster difference is large than $\sqrt{M\log n}$ for some constant $M$, the posterior distribution induced by (\ref{tprior3}) leads to satisfactory Bayesian consistency.

It is worth to mention that the $r$-update step in algorithm (\ref{gibbs3})
may critically change the stochastic stability of the algorithm and that the ergodicity of Markov chains generated by (\ref{gibbs3}) is still unclear.
Therefore, in our application we use it with caution. In all our simulations and toy examples, we initialize with $r=\widehat r$, and the $r$-update step is only implemented every other 20 iterations after certain burn-in period.

In frequentist literature, another popular approach to cluster $\theta_i$'s is 
to impose a pairwise difference penalty as $\sum_{i\neq j}p(|\theta_i-\theta_j|)$ for some penalty $p(\cdot)$\cite{ma2017concave,ShenH2012,KeFW2015}.
Although it is tempting to develop a Bayesian counterpart, i.e., using a prior of form
$\pi(\theta) \propto \prod_{i\neq j}\pi(\theta_i-\theta_j)$ with sparsity induced $\pi$, there are several problems.
First, the pairwise differences $\{\theta_i-\theta_j\}_{i\neq j}$
must satisfy triangle inequality, hence the prior specification $\pi(\theta) \propto \prod_{i\neq j}\pi(\theta_i-\theta_j)$
is actually not an independent prior over all pairwise difference $\{\theta_i-\theta_j\}_{i\neq j}$, and it will be difficult to characterize the effect of such dependency in the prior distribution.
Second, even under clustering structure, $\{\theta_i-\theta_j\}_{i\neq j}$ are actually not sparse. For example, if there are two balanced groups among $\theta_i$'s, then half of the pairwise differences are nonzero. Imposing a sparse prior on a non-sparse system
will lead to severe overshrinkage problem, and our simulation shows that such prior shrinks all $\theta_i$ towards $\bar y$. 
Third, such pairwise difference prior will also heavily increase the computational burden of posterior sampling.

To illustrate and compare the performance of (\ref{gibbs2}) and (\ref{gibbs3}), a simple toy example is conducted. 
We simulate a data set with $n=100$ data points that belong to 3 underlying subgroups. 
For both (\ref{gibbs2}) and (\ref{gibbs3}), we choose $a_t=2$, $b_t=0.005$, $a_\sigma=b_\sigma=0.5$ and $\lambda_1=5$.
In addition, we compare them with
frequentist $L_1$ fusion
\begin{equation}\label{fl1}
\widehat\theta = \arg\min\|y-\theta\|^2/2+\lambda\sum_{i=2}^n|\theta_{\widehat r(i)}-\theta_{\widehat r(i-1)}|,
\end{equation}
and Bayesian Dirichlet process modeling
\begin{equation}\label{dprior}
\begin{split}
&y_i|\theta_i,\sigma^2\sim N(\theta_i, \sigma^2),\quad \theta_i\sim G,\\
&G\sim \mbox{Dirichlet-Process}(N(0, \lambda), a)\mbox{ and }\sigma^2\sim \mbox{Inverse-Gamma}(a_\sigma,b_\sigma)
\end{split}
\end{equation}
with $\lambda=5$,  $a=0.1$ and $a_\sigma=b_\sigma=0.5$.
The posterior sampling algorithms of Dirichlet process modeling are discussed in
\cite{neal2000markov}.
The frequentist estimator and posterior boxplots of $\theta_i$s under different priors are plotted in Figure \ref{toy3}. 
This figure clearly shows that frequentist $L_1$ fusion estimator fails to yield clustering structure for monotone sorted $y_{\widehat r(i)}$'s, especially for the middle portion of the data, one always have $\widehat\theta_{\widehat r(i)}=y_{(i)}$.
Dirichlet process induces a quite reasonable posterior clustering, however its posterior concentration is not good, in the sense that the posterior variance is quite large.
In the opposite, the $t$ modeling (\ref{tprior2}) has a strong prior concentration, which contributes to a better posterior variation, and
 encourages posterior clustering. 
However, consistent to our previous arguments, a severe over-clustering
occurs, and the posterior identifies more than 6 clusters.
At last, the posterior obtained by algorithm (\ref{gibbs3}) combines both advantages of DP prior and (\ref{tprior2}). Comparing with DP prior, it
has much smaller posterior variance; comparing with prior (\ref{tprior2}),
the over-clustering issue is greatly alleviated.

\begin{figure}[htb]
	\begin{center}
		\includegraphics[width=12cm]{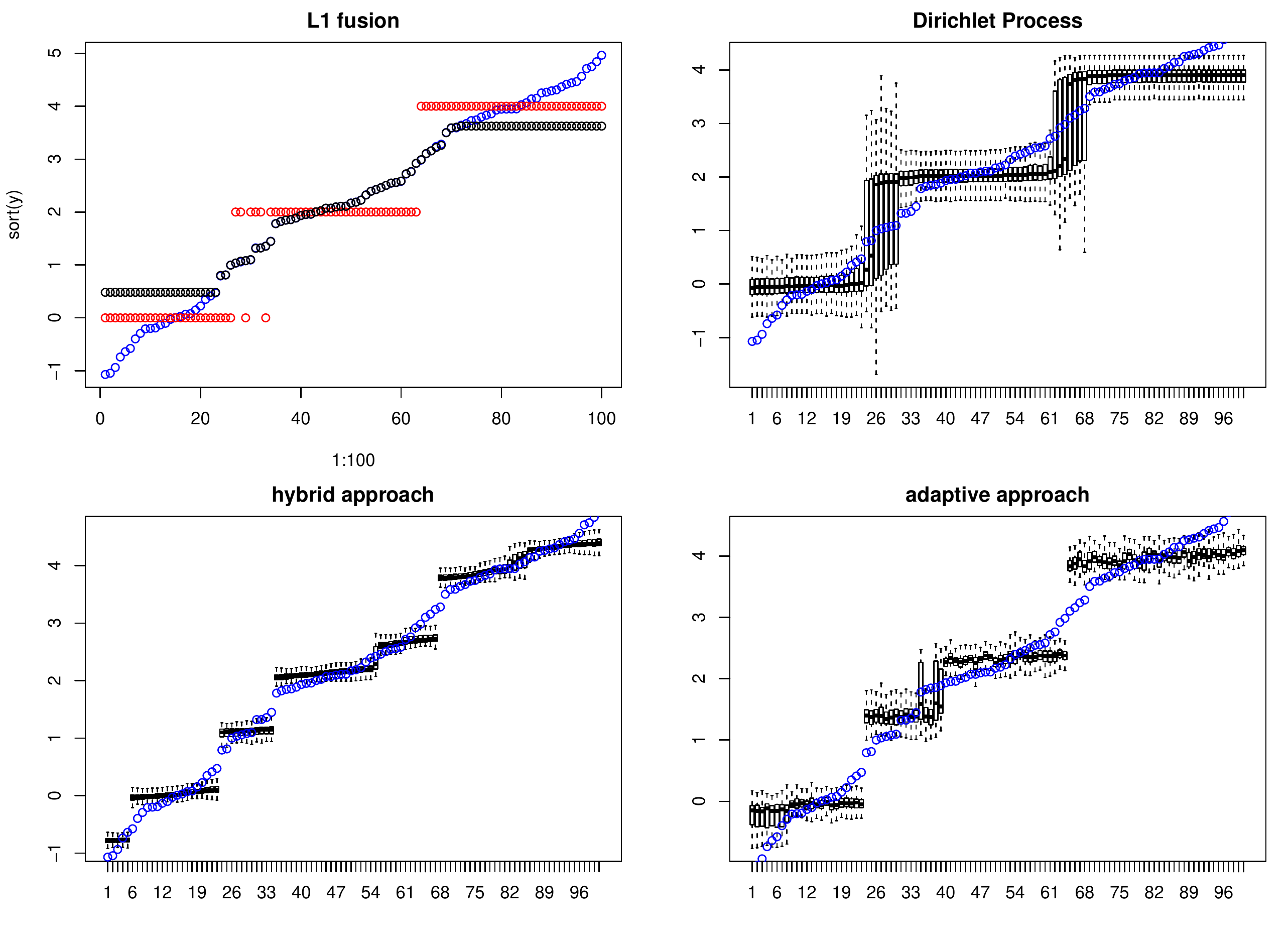}\vskip -0.3in
		\caption{
			Toy example comparison among difference approaches.
			Upper Left: Frequentist fusion estimation (\ref{fl1});
			Upper right: Marginal posterior boxplots for each $\theta_i$ under Laplace process prior; Lower left: Marginal posterior boxplots for each $\theta_i$
			under t prior (\ref{gibbs2}); Lower right: Marginal posterior boxplots under $t$ prior (\ref{gibbs3}).
			The blue points denote the sorted observed data $y_{(i)}$.
			The red  points denote the true $\theta^*$ corresponding to the data ascending order.
		}\label{toy3}
	\end{center}
\end{figure}

\section{Simulation and data anlaysis}\label{simu}

\subsection{ Bayesian fusion simulations}\label{simu1}

In our first simulation studies, we consider model 
(\ref{model}) with $\sigma^*=0.5$ and true parameter $\theta^*\in \BR^{100}$ having three consecutive blocks. These three blocks are
$\mB_1^* =\{1,\dots, b_1\}$, $\mB_{2}^*=\{b_1+1,\dots, b_1+b_2\}$, 
$\mB_{3}^*=\{b_1+b+2+1,\dots, b_1+b_2+b_3\}$ where
$(b_1,b_2,b_3)\sim \mbox{multinomial}(100, (1/3,1/3,1/3))$,
and $\theta_i^* =0$, 2 or 4 within each block respectively.
We compare the performance among Bayesian $t$ fusion (\ref{tprior}), Bayesian Laplace fusion (\ref{Lprior})
and frequentist $L_1$ fusion (\ref{fusest}), based on 100 replications. The choices of the hyperparameters are same to the toy example discussed in Section \ref{sectoy}.

To compare the accuracy of parameter estimations, we calculate the $L_2$ and $L_1$ estimation errors, $\|\widehat\theta-\theta^*\|_2$ and
$\|\widehat\theta-\theta^*\|_1$, for frequentist and Bayes estimator (posterior mean).
To assess the performance of block structure recovery, we consider several 
comparison criteria. We define the  ``adjacency" matrix as $\Sigma=(\omega_{ij})=(1\{\theta_i=\theta_j\})$, and corresponding 
estimation error of $\Sigma$: $R=\|\Sigma^*-\widehat\Sigma\|_1=\sum_{i,j}
|\widehat\omega_{ij}-\omega_{ij}^*|$.
For $L_1$ fusion estimator, $\widehat\Sigma$ is trivially estimated by
$\widehat\omega_{ij}=1\{\widehat\theta^{\rm F}_i=\widehat\theta^{\rm F}_j\}$.
For Bayesian shrinkage approaches, since their posterior samples are continuous, it is necessary to ``sparsify'' the continuous posterior in order 
to retrieve a sparse structure estimation for $\Omega$.
In the literature, this is usually done by 1) hard thresholding
\cite{LiP2017,TangXGG2016, CarvalhoPS2010, IshawaranR2005}, 
or 2) decoupling shrinkage and selection methods \cite{HahnC2015,XuSMQH2017}.
All these mentioned approaches depend solely on the posterior scaling, and donot take into account of the degree of prior concentration. Therefore, following the suggestion made in \cite{SongL2017}, we estimate
$\widehat\omega_{ij}=1\{|\widehat\theta_i-\widehat\theta_j|/\widehat\sigma
\leq \pi_{1/2n}\}$, where
$\widehat\theta_i$ and $\widehat\sigma$ are the Bayes estimator, 
$\pi_{1/2n}$ denotes the $(1-1/2n)$ quantile of the
prior distribution imposed on successive difference $\vartheta_i/\sigma$.
This choice tries to increases the robustness of Bayesian inference: if an
overly small or large scale parameter $s$ is used, the estimation for $\omega_{i,j}$ will adapt accordingly.
For $t$ shrinkage (\ref{tprior}), $\pi_{1/2n}=st_{1/2n}$ where $s$ is the scale parameter, $t_{1/2n}$
is the $1-1/2n$ quantile of $t$ distribution with df$=2a_t$;
for Laplace shrinkage (\ref{Lprior}), $\pi_{1/2n}=\log(n)/\lambda$.
Since different priors have different corresponding $\pi_{1/2n}$ values, and comparison solely based on the value of $R$ may not be completely fair. Hence, we also consider the following two statistics:
$W=\mbox{Average}_{\{w_{ij}^*\neq0\}}|\widehat \theta_i-\widehat\theta_j|$ denotes within-block average variation, and
$B=\min_{\{w_{ij}^*=0\}}|\widehat \theta_i-\widehat\theta_j|$ denotes 
the between-block separation.
A larger $B$ and smaller $W$ indicate a better block identification performance. 
The results are summarized in Table \ref{T1}.

\begin{table}[ht]
\caption{Comparison among Bayesian t-fusion, Bayesian Laplace fusion, and $L_1$ fusion. 
The report are based on average results from 100 replications. Refer to Section \ref{simu1} for the definition of $R$, $W$ and $B$.}\label{T1}
\begin{center}
\begin{tabular}{cccc}
 \hline
Methods                          & Bayesian t-fusion   & Bayesian Laplace fusion    &$L_1$ fusion\\\hline
$L_2$ error of $\theta$          &1.3243           &2.2323            &1.5916  \\ 
Standard error                   &0.0602           &0.0324            &0.0421  \\ \hline
$L_1$ error of $\theta$          &9.6964           &16.6296           &11.6745  \\ 
Standard error                   &0.3419           &0.2725            &0.3407  \\\hline
\hline
$W$ (the smaller the better)     &0.05584          &0.2182            &0.2128 \\ 
Standard error                   &0.0028           &0.0037            &0.0077\\\hline
$B$ (the larger the better)     &1.4243           &0.6792            &1.1302 \\ 
Standard error                   &0.0684           &0.0211            &0.0398 \\\hline
$R$ (the smaller the better)     &387.32           &85.6              &1360.9    \\ 
Standard error                   &32.0628          &6.2375            &53.2685    \\\hline
\hline                                                     
\end{tabular}
\end{center}
\end{table}

The simulation results show that the Bayesian $t$ fusion yields the most accurate estimation in terms of both $L_2$ and $L_1$ errors. This Bayes estimator also induces the best clustering result, as it has the largest between-group separation and smallest within-group variation.
In comparison, the frequentist fusion estimator has a worse accuracy. It is well known that $L_1$ penalty introduces estimation bias
\cite{FanL2001}, but unlike the LASSO penalty which introduces a bias of $\lambda$, fuse-$L_1$ penalty introduces only a bias as small as $\lambda/n_i$ (\cite{rinaldo2009properties}) where $n_i$
is the number of elements in each block.
Hence, its estimation performance is still acceptable.
The $L_1$ fusion also has a much larger $R$ statistics. This is 
consistent to observation from our toy example, that $L_1$ fusion has
a mild over-partition issue.
Bayesian Laplace fusion, on the other hand, has a much worse estimation behavior due to its insufficient shrinkage effect. As discussed in Section 
\ref{sectoy}, the Laplace fusion tends to yield smooth changing $\widehat\theta_i$'s, thus it has a much smaller between-block separation.
Note that Laplace shrinkage does obtain the smallest $R$ statistic,
but this doesn't imply that it has a good performance on structure recovery.
It has a small $R$ statistics because the Laplace-fusion prior doesn't have strong prior concentration and its corresponding $\pi_{1/2n}$ is much larger than $t$ fusion. 

\subsection{ Bayesian clustering simulations}\label{simu2}
In our second simulation studies, we consider model 
(\ref{model}) with $\sigma^*=0.5$ and true parameter $\theta^*\in \BR^{100}$ having three unknown clusters.
With equal probability, $\theta_i^*=0$, 2, or 4. 
We aim to compare different approaches including Bayesian $t$ modeling (\ref{gibbs2}) and (\ref{gibbs3}),
Dirichlet process prior (\ref{dprior}) and the frequentist bCARDS estimation (\ref{fl1}) using $L_1$ fusion \cite{ke2015homogeneity}.
Besides the comparison of $L_2$ and $L_1$ errors of the Bayesian/frequentist estimator, we also compare the posterior mean of squared $L_2$ error,
$E_{\pi(\theta|y)}\|\theta-\theta^*\|^2$, among the Bayesian approaches.
Note that the $L_2$ error of Bayes estimator only tells how good is the posterior center, while the posterior mean of squared $L_2$ error tells how 
good is distributional posterior contraction.

Note that the simulation setting seems similar to our previous study in Section \ref{simu1}, but as mentioned in Section \ref{secada}, clustering problem is much more difficult than fusion problem. 
To see that, assume there are three balanced clusters for $\theta$, let $z_1$ be the largest
observation in the 0's cluster and $z_2$ be the smallest observation in the 2's cluster.
By the well known result \cite{royston1982algorithm}, $E(z_1)\approx 1.04 > 0.96 \approx E(z_2)$. Therefore, misclassification will always happen for those extreme observations, and 
$B=\min_{\{w_{ij}^*=0\}}|\widehat \theta_i-\widehat\theta_j|$ is always around 0 for all methods, therefore, the comparison of $B$ statistic values are meaningless.
Because of that, to assess how well these methods identify and separate unknown clusters, we re-define the $B$ statistic as 
\[\tilde B=\mbox{the bottom 10\% quantile of }\{|\widehat \theta_i-\widehat\theta_j|\}_{\{w_{ij}^*=0\}}.\]
The simulation results are summarized in Table \ref{T2}.

\begin{table}[ht]
\caption{Comparison among Bayesian t-shrinkage, Bayesian DP prior, and $L_1$ fusion. 
The report are based on average results from 100 replications. Refer to Section \ref{simu2} for the definition of $R$, $W$ and $\tilde B$.}\label{T2}
\begin{center}
\begin{tabular}{cccccc}
 \hline
Methods                          &$t$-shrinkage (\ref{gibbs3})  & $t$-shrinkage (\ref{gibbs2})  
&DP prior(\ref{dprior}) &$L_1$ fusion (\ref{fl1}) \\\hline
posterior mean of $L_2^2$ error         &20.559           &20.5660      & 40.7029   &--- \\ 
Standard error                   &0.6885           &0.4589       & 3.1209    &---  \\ \hline
$L_2$ error of $\theta$          &4.1873           &4.3925       &  3.9584   &5.0039  \\ 
Standard error                   &0.0872           &0.0531       &  0.1621   &0.0739  \\ \hline
$L_1$ error of $\theta$          &28.1099          &33.4543      & 26.8122   &43.8450  \\ 
Standard error                   &0.7385           &0.4845       &  1.7122   &0.7909  \\\hline
\hline
$W$ (the smaller the better)     &0.3517           &0.4428       & 0.2183    &0.4551 \\ 
Standard error                   &0.0093           &0.0063       & 0.0061    &0.0062\\\hline
$\tilde B$ (the larger the better)      &1.4652           &1.3392       & 1.3885    &0.9855 \\ 
Standard error                   &0.0218           &0.0147       & 0.0419    &0.0146 \\\hline
$R$ (the smaller the better)     &1752.24          &2294.76      &  ---      &1973.7   \\ 
Standard error                   &31.5641          &14.6102      &  ---      &11.9937   \\\hline
\hline                                                     
\end{tabular}
\end{center}
\end{table}

The comparison shows that DP prior yields the best point estimator in terms of estimation error, and the adaptive $t$-shrinkage method (\ref{gibbs3}) gives a slightly worse point estimator. 
But in terms of the posterior contraction, both $t$-shrinkage approaches have much smaller posterior mean of squared $L_2$ error.
This result is consistent to the insight obtained from the toy example in Section \ref{secada} that LP prior leads to larger posterior variance.
The three Bayesian approaches have approximately the same performance for between-cluster
separation ($\tilde B$ statistic), while the DP prior has a smaller
within-cluster variation ($W$ statistic).
The frequentist estimator (\ref{fl1}) has the worst performance in almost every aspect.
In summary, DP prior does yield the best Bayes point estimation, but the adaptive $t$-shrinkage (\ref{gibbs3}) has the best Bayesian contraction, and 
its posterior mean, as a point estimator, has reasonable performance for estimation and  clustering structure recovery.
Furthermore, by comparing the results between $t$-shrinkage method (\ref{gibbs3}) and (\ref{gibbs2}), we conclude that the $r$-update step implemented in algorithm (\ref{gibbs3})
does improve the performance of $t$-shrinkage in every aspect for this clustering problem.

\subsection{ Real data set analysis}\label{simu3}
In this section, we consider a real comparative genomic hybridization
(CGH) dataset \cite{tibshirani2007spatial}. The dataset is available from the R package \verb$cghFLasso$ and it contains n=990 observations. 
We are interested in a fusion estimation for the mean trend in order to detect the “hot-spot” region, i.e., the corresponding genes have extra DNA copies.
We apply both Bayesian $t$-shrinkage fusion (\ref{tprior}) and frequentist
$L_1$ fusion (\ref{fusest}).
For the Bayesian application, the hyperparameter is chosen as 
$a_t=2$ and $b_t=0.0007$, following the suggestion of (\ref{tune}).
Both point estimations are plotted in Figure \ref{real}.
Frankly speaking, the two estimations share similar pattern, and it is difficult to judge which one is definitely better.
To see the slight advantage of $t$-shrinkage, a closer comparison is displayed in the same figure as well, where we compare the two estimators within the segment [80,120].
It is obvious that both approaches induce exactly the same blocking structure. But the difference is that, for $t$-shrinkage prior, the $\widehat\theta_i$ within each block is very close to the sample mean of the block; but for $L_1$ fusion, there is a large bias between 
$\widehat\theta_i$ and the sample mean of the block.

\begin{figure}[htb]
	\begin{center}
		\includegraphics[width=12cm]{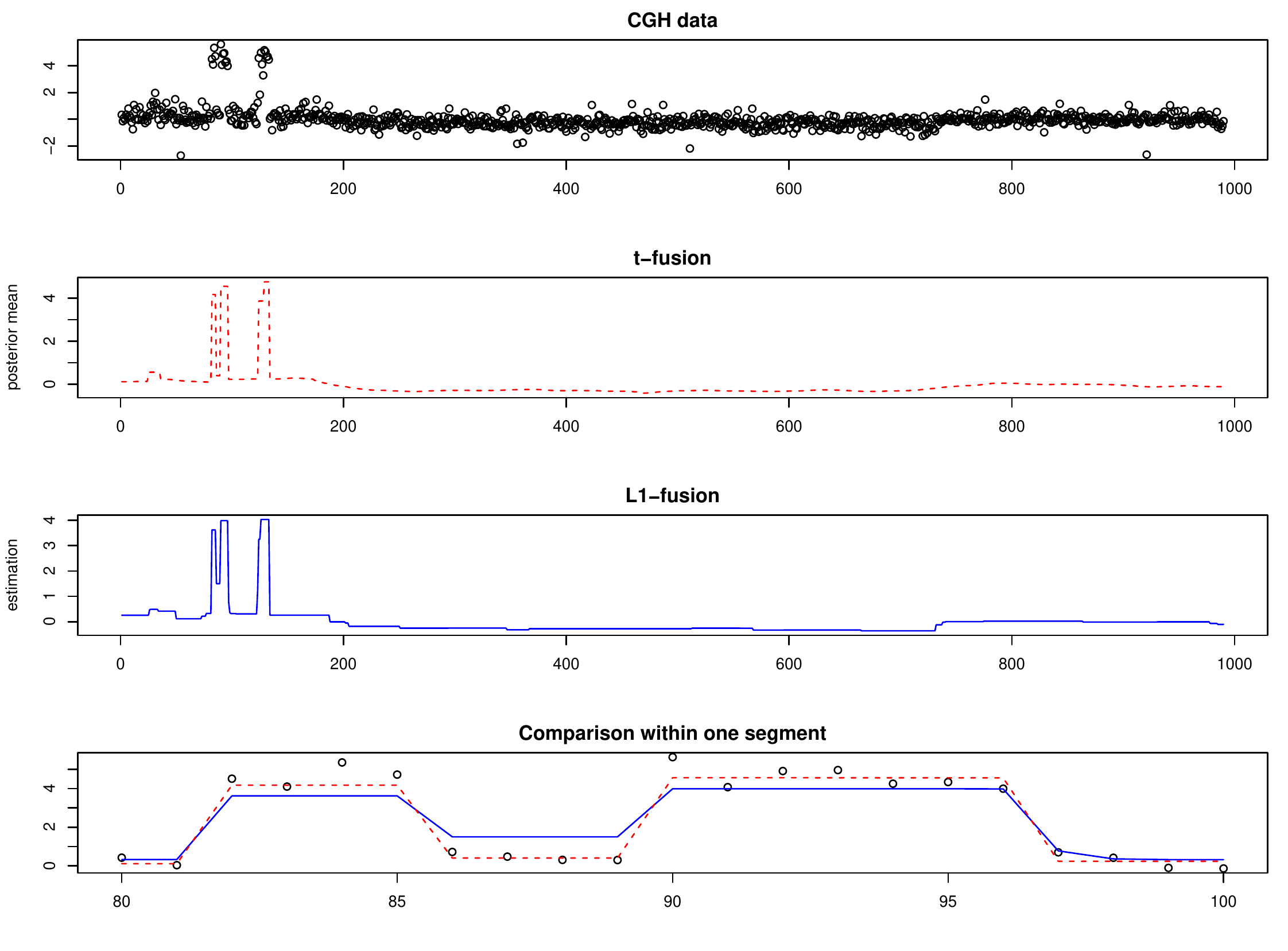}\vskip -0.3in
		\caption{
			Real data application: 1) CGH data; 2) posterior mean of $t$
			shrinkage fusion; 3) frequentist $L_1$ fusion estimator;
			4) Comparison of the two approaches within a segment of [80,120].
		}\label{real}
	\end{center}
\end{figure}

\section{Final remarks}\label{end}
In this work, we study the Bayesian inference for a vector parameter $\theta$ which has a unknown blocking structure, using Bayesian shrinkage prior. 
For the ease of representation, this work focuses on the application of $t$-fusion prior, but our theorems actually holds for any polynomially decaying prior distribution. We demonstrate that a simple $t$-fusion prior
leads to satisfactory posterior contraction, and is powerful to recover the blocking structure. We recommend not to use 
Laplace fusion prior since it can only induce a smoothly varying posterior
estimation.
Although this work mainly focuses on the normal sequence models, but
the presented $t$-fusion modeling and the posterior asymptotic results can be easily extended to more complicated regression models such as $y=X\beta+Z\theta+\epsilon$ for some blocky parameter $\theta$. 

We also extend the use of shrinkage prior to a more general clustering problem. To the best of authors' knowledge, this is the first attempt in literature to use Bayesian shrinkage to recover unknown clustering structure.
The basic idea is to find a pilot order $\widehat r$, and then fuse all pairs of neighbors (that are determined by $\widehat r$) via shrinkage priors. An adaptive $r$-update step is further incorporated to improve the clustering performance.
Simulations show that the proposed algorithms (\ref{gibbs2}) and (\ref{gibbs3}) have reasonable performance, and their posterior contraction is much better than conventional Dirichlet process modeling. Further theoretical investigation is necessary to understand their asymptotic behaviors. 
In practice, this idea can be generalized to clustering problem of multi-dimensional data as well, i.e., $y_i=\theta_i\in\BR^d$.
Since we can not rank multivariate vectors, the possible alternative is to construct a minimum spanning tree (MST)\cite{li2018spatial}, and impose shrinkage fusion prior on the pair of $\theta_i$ that is connected by a edge in the MST.

\section*{Acknowledgement}
This work is in memory of Prof. Jayanta Ghosh who has jointly supervised the first PhD student of the second author.
Qifan Song's research is sponsored by NSF DMS-1811812.
Guang Cheng's research is sponsored by NSF DMS-1712907, DMS-1811812, DMS-1821183, and
Office of Naval Research, (ONR N00014-18-2759).

\appendix
\section{Appendix}
First, let us state some useful lemmas.
\begin{lemma}[Lemma 1 of \citep{LaurentM2000}]\label{chi}
	Let $\chi^2_{d}(\kappa)$ be a chi-square distribution with degree of freedom $d$, and noncentral parameter $\kappa$, then we have the following concentration inequality
	\[\begin{split}
	&Pr( \chi^2_{d}(\kappa)>d+\kappa+2x+\sqrt{(4d+8\kappa)x})\leq \exp(-x), \mbox{ and}\\
	&Pr( \chi^2_{d}(\kappa)<d+\kappa-\sqrt{(4d+8\kappa)x})\leq \exp(-x).
	\end{split}\]
\end{lemma}

\begin{lemma}[Theorem 1 of \citep{ZubkovS2013}]\label{lemmad}
        Let $X$ be a Binomial random variable $X\sim \mbox{B}(n,v)$. For any $1<k<n-1$
	\[ 
	Pr(X\geq k+1)\leq 1- \Phi(\mbox{sign}(k-nv)\{2nH(v, k/n)\}^{1/2}),
	\]
	where $\Phi$ is the cumulative distribution function of standard Gaussian distribution and
	$H(v, k/n)= (k/n)\log(k/nv)+(1-k/n)\log[(1-k/n)/(1-v)]$.
\end{lemma}

The next lemma is a refined result of Lemma 6 in \citep{Barron1998}:
\begin{lemma}\label{lemmab}
	Let $f^*$ be the true probability density of data generation, $f_{\theta}$ be the likelihood function
	with parameter $\theta\in\Theta$, and $E^*$, $E_\theta$ denote the corresponding expectation
	respectively. Let $B_n$ and $C_n$ be two subsets in parameter space $\Theta$, and $\phi_n$ be some
	test function satisfying 
	$\phi_n(D_n)\in[0,1]$ for any data $D_n$. If $\pi(B_n)\leq b_n$, $E^*\phi(D_n)\leq b_n'$,
	$\sup_{\theta\in C_n}E_{\theta}(1-\phi(D_n))\leq c_n$, and  furthermore,  
	\[P^*\left\{\frac{m(D_n)}{f^*(D_n)}\geq a_n \right\}\geq 1-a_n',\]
	where $m(D_n)=\int_{\Theta} \pi(\theta)f_{\theta}(D_n)d\theta$ is the margin probability of $D_n$. Then,
	\[
	\begin{split}
	&E^*\left(\pi(C_n\cup B_n)|D_n)\right)\leq \frac{b_n+c_n}{a_n}+a_n'+b_n'.
	\end{split}
	\]
\end{lemma}
\begin{proof}
	Define $\Omega_n$ to be the event of $(m(D_n))/(f^*(D_n))\geq a_n$,
	and $m(D_n, C_n\cup B_n) = \int_{C_n\cup B_n} \pi(\theta)f_{\theta}(D_n)d\theta$. Then
	\[
	\begin{split}
	&E^*\pi(C_n\cup B_n)|D_n) = E^*\pi(C_n\cup B_n)|D_n)(1-\phi(D_n))1_{\Omega_n}\\
	+&E^*\pi(C_n\cup B_n)|D_n)(1-\phi(D_n))(1-1_{\Omega_n})+ E^*\pi(C_n\cup B_n)|D_n)\phi(D_n)\\
	\leq& E^*\pi(C_n\cup B_n)|D_n)(1-\phi(D_n))1_{\Omega_n}+E^*(1-1_{\Omega_n})+ E^*\phi(D_n)\\
	\leq& E^*\pi(C_n\cup B_n)|D_n)(1-\phi(D_n))1_{\Omega_n}+b_n'+a_n'\\
	\leq& E^*\{m(D_n, C_n\cup B_n)/a_nf^*(D_n)\}(1-\phi(D_n))+b_n'+a_n'.
	\end{split}
	\]
	
	By Fubini theorem,
	\[\begin{split}
	&E^*(1-\phi(D_n))m(D_n, C_n\cup B_n)/f^*(D_n) = \int_{C_n\cup B_n} \int_{\mX}[1-\phi(D_n)]f_\theta(D_n) dD_n\pi(\theta)d\theta\\
	\leq&\int_{C_n} E_{\theta}(1-\phi(D_n))\pi(\theta)d\theta+\int_{ B_n}  \int_{\mX}f_\theta(D_n) dD_n\pi(\theta)d\theta\leq b_n+c_n.
	\end{split}\]
	
	Combining the above inequalities leads to the conclusion.
\end{proof}
{\noindent \bf Proof of Theorem \ref{thm1} and \ref{thm2}}
\vskip 0.1in
Let $G=\{g_1,g_2,\dots,g_d\}$ be a generic subset of $\{2,\dots, n\}$, and it also represents potential a $(d+1)$-group structure of $\theta$
as $\{\{1,\dots, g_1\},\{ g_1+1,\dots, g_2\}, \dots, \{g_{d}+1,\dots, n \}\}$.
Given $G$ and its corresponding blocking structure, $\widehat\theta_G(y)$ denotes the estimation of $\theta$ based on block mean, i.e. $\widehat\theta_{G,j}(y) = \sum_{i=g_j+1}^{g_{j+1}}y_i/(g_{j+1}-g_j)$ for all $g_j+1\leq j\leq g_{j+1}$,
and $\widehat{\sigma}^2_G(y)=\|y-\widehat\theta_G\|^2/(n-|G|-1)$.

To prove the posterior contraction, we will apply Lemma \ref{lemmab}. We define the following testing function
\begin{equation}\begin{split}
    \phi(y) = 1\{&\|\widehat\theta_G-\theta^*\|\geq \sqrt{n}\sigma^*\epsilon_n\mbox{ and }
    |\widehat{\sigma}^2_G-\sigma^{*2}|>\sigma^{*2}\epsilon_n\\&\mbox
    { for all }G\supset G^*,
    |G|\leq (1+\delta)|G^*|\}
    \end{split}
\end{equation}
for some $\delta>0$, and define $C_n$ and $B_n$ as:
\[\begin{split}
&C_n=\{\theta: \|\theta-\theta^*\|\leq M\sqrt n\sigma^*\epsilon_n, (1-\epsilon_n)/(1+\epsilon_n)<\sigma^2/\sigma^{*2}<(1+\epsilon_n)/(1-\epsilon_n)\}^c\backslash B_n,\\
&B_n=\{\theta: \mbox{Among all }\{\vartheta_i's\}_{i\notin G^*} , \mbox{ there are at least }\delta|G^*|
 \mbox{ of them are greater than }\sigma\epsilon_n/n\}.
\end{split}\]

Note that when $G\supset G^*$, $\|\widehat\theta_G(y)-\theta^*\|^2\sim \sigma^{*2}\chi^2_{|G|+1}$ and $\|y-\widehat\theta_G\|^2\sim \sigma^{*2}\chi^2_{n-|G|-1}$, thus by Lemma \ref{chi},
we have that 
\[
P(\|\widehat\theta_G(y)-\theta^*\|\geq \sqrt{n}\sigma^*\epsilon_n\mbox{ and }
|\widehat{\sigma}^2_G-\sigma^{*2}|>\sigma^{*2}\epsilon_n)
\leq \exp\{-c_1n\epsilon_n^2\},
\]
for some constant $c_1$, since $|G|=O(|G^*|)\prec n\epsilon_n^2$ and $\epsilon_n\prec 1$.
Therefore,
\begin{equation}\label{test1}
    E_{(\theta^*,\sigma^{*2})}\phi(y)\leq {n-1 \choose (1+\delta)|G^*|}\exp\{-c_1n\epsilon_n^2\}=\exp\{-c_1'n\epsilon_n^2\},
\end{equation}
as long as $n\epsilon^2/[|G^*|\log n]$ is sufficiently large.

For any $(\theta,\sigma^2)\in C_n$ satisfying $\|\theta-\theta^*\|\leq M\sqrt n\sigma^*\epsilon_n$ and 
$\sigma^2/\sigma^{*2}\geq(1-\epsilon_n)/(1+\epsilon_n)$, we define $\widehat G=\{i:\theta_i-\theta_{i-1}\geq \sigma\epsilon_n/n^2\}\cup G^*$ ( hence $|G|\leq (1+\delta)|G^*|$), thus
\[
\begin{split}
 & P_{(\theta,\sigma^2)}(  \|\widehat\theta_{\widehat G}(y)-\theta^*\|\leq \sqrt{n}\sigma^*\epsilon_n)
  = P_{(\theta,\sigma^2)}(  \|\widehat\theta_{\widehat G}(y)-\widehat\theta_{\widehat G}(\theta) + \widehat\theta_{\widehat G}(\theta)-\theta^*\|\leq \sqrt{n}\sigma^*\epsilon_n)\\
 \leq&P_{(\theta,\sigma^2)}(\|\widehat\theta_{\widehat G}(y)-\widehat\theta_{\widehat G}(\theta)\| \geq \| \widehat\theta_{\widehat G}(\theta)-\theta^*\|-\sqrt{n}\sigma^*\epsilon_n) \\
 \leq& P_{(\theta,\sigma^2)}(\|\widehat\theta_{\widehat G}(y)-\widehat\theta_{\widehat G}(\theta)\| \geq M\sqrt{n}\sigma^*\epsilon_n-\sqrt{n}\sigma\epsilon_n-\sqrt{n}\sigma^*\epsilon_n) \\
 \leq&P\left(\chi_{|G+1|}^2\geq \left[\sqrt{\frac{1-\epsilon_n}{1+\epsilon_n}} (M-1)-1\right] n\epsilon_n^2\right)\leq \exp\{-c_2n\epsilon_n^2\}
\end{split}
\]
for some $c_2$ given a large $M$, where the second inequality is due to the fact that $\|\widehat\theta_{\widehat G}(\theta)-\theta\|\leq \sqrt{n}\sigma\epsilon_n$ when $\theta\in B_n$. 

For any $(\theta,\sigma^2)\in C_n$ satisfying $\sigma^2/\sigma^{*2}<(1-\epsilon_n)/(1+\epsilon_n)$
or $\sigma^2/\sigma^{*2}>(1+\epsilon_n)/(1-\epsilon_n)$,
\[
\begin{split}
 & P_{(\theta,\sigma^2)}(   |\widehat{\sigma}^2_G-\sigma^{*2}|<\sigma^{*2}\epsilon_n)
  = P_{(\theta,\sigma^2)}(   |\|y-\widehat\theta_G\|^2/\sigma^{*2}(n-|G|-1)-1|<\epsilon_n)\\
  \leq & P_{(\theta,\sigma^2)}( 1-\epsilon_n<  \|y-\widehat\theta_G\|^2/\sigma^{*2}(n-|G|-1)<1+\epsilon_n)\\
  \leq &P_{(\theta,\sigma^2)}\left( \bigg|\frac{\|y-\widehat\theta_G(y)\|^2}{\sigma^{2}}-(n-|G|-1)\bigg|>(n-|G|-1)\epsilon_n\right)\\
  =&P_{(\theta,\sigma^2)}\left( |\chi_{n-|G|-1}^2(\lambda)-(n-|G|-1)|>(n-|G|-1)\epsilon_n\right)\leq \exp\{-c_2'n\epsilon_n^2\}
\end{split}
\] for some $c_2'$, 
where the noncentral parameter $\lambda<n\epsilon_n^2\prec(n-|G|-1)\epsilon_n$.

Combining the results from the previous two paragraph, it is easy to obtain that 
\begin{equation}\label{test2}
     \sup_{(\theta,\sigma^{2})\in C_n}E_{(\theta,\sigma^{2})}[1-\phi(y)]\leq 
    \max\{\exp(-c_2n\epsilon_n^2),\exp(-c_2'n\epsilon_n^2)\}.
\end{equation}

Now we consider the marginal posterior density of data $y$. With probability $P(\|\varepsilon\|\leq2\sqrt n\sigma^*)$ (which converges to 1), 
\[\begin{split}
&\frac{m(y)}{f^*(y)}=\frac{\int_{\sigma^2}\int_\theta
\sigma^{*n}\exp\{-\|\theta^*-\theta+\varepsilon\|^2/\sigma^{2}\}
d\theta d\sigma^2
}{\sigma^{n}\exp\{-\|\varepsilon\|^2/\sigma^{*2}\}}\\
\geq&\int_{\sigma^2}\int_\theta\exp\left\{-\frac{\|\theta^*-\theta\|^2}{\sigma^2}
-2\frac{(\theta^*-\theta)^T\varepsilon}{\sigma^2}+\frac{\|\varepsilon\|^2}{\sigma^{*2}}-\frac{\|\varepsilon\|^2}{\sigma^{2}}-n\log(\sigma/\sigma^*)\right\}
\pi(\theta,\sigma^2)d\theta d\sigma^2\\
\geq& \pi(\max\{|\theta_1-\theta_1^*|, |\vartheta_i-\vartheta_i^*|\}/\sigma\leq  |G^*|\log n/n^2, 0\leq\sigma^2-\sigma^{*2}\leq \sigma^{*2}|G^*|\log n/n)\exp\{-c_3'|G^*|\log n \}
\end{split}\]
for some constant $c_3'$. Besides,
\[\begin{split}
&\pi(\max\{|\theta_1-\theta_1^*|, |\vartheta_i-\vartheta_i^*|\}/\sigma\leq  |G^*|\log n/n^2, 0\leq\sigma^2-\sigma^{*2}\leq\sigma^{*2} |G^*|\log n/n)\\
\geq&\pi_\sigma(\sigma^{*2})*O(\sigma^{*2}|G^*|\log n/n)* \underline\pi_\theta*O( |G^*|\log n/n^2)  \\
&*\underline\pi_\vartheta^{|G^*|} [ |G^*|\log n/n^2)]^{|G^*|} * [\pi_\vartheta(\{-|G^*|\log n/n^2, |G^*|\log n/n^2\})]^{n}\\
=&\exp\{-c_3''|G^*|\log n\}. \quad\mbox{ (by the conditions imposed on the prior specifications)}
\end{split}
\]for some constant $c_3''$. 
Thus 
\begin{equation}\label{md}
    m(y)/f^*(y)\geq \exp\{-(c_3'+c_3'')|G^*|\log n\}
    =\exp\{-c_3n\epsilon_n^2\}, \mbox{ with probability tending to 1}
\end{equation}
for some $c_3$, where $c_3$ can be sufficiently small when $n\epsilon_n^2/[|G^*|\log n]$ is large enough.

At last, we study the prior probability of set $B_n$. Due to the prior independence of $\vartheta_i'$s,
$\pi(B_n) = \pi[\mbox{Bin}(n-1-|G^*|,p)>(\delta)|G^*| ]$, where $p\leq (1/n)^{1+u}$. By Lemma \ref{lemmad},
\begin{equation}\label{pib}
    \pi(B_n) \leq \exp\{ -c_4\delta |G^*|\log n\}
\end{equation}
 for some $c_4$.
Combine results (\ref{test1}), (\ref{test2}), (\ref{md}) and (\ref{pib}), and we apply Lemma \ref{lemmab}
to get the posterior consistency result that 
\[
\pi(B_n\cup C_n|y) \rightarrow^p 0,
\]
given a sufficient large constants $\delta$ and $n\epsilon_n/|G^*|\log n$.
 
 \vskip 0.2in
{\noindent\bf Proof of Theorem \ref{pos1}}

Consider $y=\theta+d$, where $\theta_i^*\equiv 0$ for all $i$ and error $d$ is order statistics of standard normal variables, i.e. the density of $d$ is $f(d)=n!\prod \phi(d_i)1(d_1\leq d_2,\dots,d_{n-1}\leq d_n)$ and $\phi$ denotes the standard normal density. The prior of $\theta$ follows
$\pi(\theta)=\pi_1(\theta_1)\prod_{i=2}^{n}\pi_{t,s}(\theta_i-\theta_{i-1})$,
where $\pi_{\lambda_1}$ is the density of $N(0,\lambda_1)$, and $\pi_{t,s}$ is the density of $t$ distribution with tiny scale parameter satisfying $-\log s\asymp \log n$, i.e. conditions in Corollary \ref{thmt} holds, and we consider the misspecified posterior of form
$\pi(\theta|D_n)=\exp\{-(y-\theta)^2/2\}\pi(\theta)$.

Define $\mu\in\BR^n$ as $\mu_i=0$ for all $1\leq i\leq k= 3n/4$, and $\mu_i=Z_{0.25}/2$ for $ i> k$ where $Z_{0.25}$ is the right 25\% quantile of standard normal distribution, thus $\|\mu-\theta^*\|^2\asymp n$.

Let $\Delta\theta$ be any vector such that $\|\Delta\theta\|^2\leq M\log n$.
Then 
\[
-\log\left(\frac{\pi(\mu+\Delta\theta)}{\pi(\theta^*+\Delta\theta)}\right)
=-\log\left(\frac{\pi_{t,s}(Z_{0.25}/2+\Delta\theta_{k})}{\pi_{t,s}(\Delta\theta_{k})}\right) = O( -\log s) = O(\log n).
\]
And
\[\begin{split}
&\log\left(\frac{\exp\{-(y-\mu-\Delta\theta)^2/2\}}{\exp\{-(y-\theta^*-\Delta\theta)^2/2\}}\right)
=[(y-\theta^*-\Delta\theta)^2-(y-\mu-\Delta\theta)^2]/2\\
=&\frac{1}{2}\sum_{i=k+1}^n[(y_i-\Delta\theta_i)^2-(y_i-Z_{0.25}/2-\Delta\theta_i)^2]=\frac{1}{2}\sum_{i=k+1}^n[(y_i-\Delta\theta_i)Z_{0.25}-\frac{Z_{0.25}^2}{4}]\\
\geq &\frac{1}{2}\left[(\|y_{k+1:n}\|_1-\sqrt{nM\log n/4})Z_{0.25}-\frac{nZ_{0.25}^2}{16}  \right]\geq cn,
\end{split}
\]
for some positive constant $c$ given sufficiently large $n$, where the inequalities above hold
since $y_n\geq y_{n-1}\dots\geq y_{k+1}$, and $y_{k+1}\approx Z_{0.25}$ with high probability, due to large sample empirical quantile theory.

Combining the above two results, we have that the posterior density satisfies 
$\pi(\mu+\Delta\theta|D_n)\gg \pi(\theta^*+\Delta\theta|D_n)$ for any $\|\Delta\theta\|^2\leq M\log n$ with high probability. Therefore, more posterior mass is distributed within the $\sqrt{M\log n}$-radius ball centered at $\mu$ than at the true parameter $\theta^*$.
 
 \vskip 0.2in
{\noindent\bf Proof of Theorem \ref{thm3}}

 The proof of this theorem is quite similar to the proof of Theorem \ref{thm1} and \ref{thm2}.  We define the same testing function
 as in the proof of theorem \ref{thm1} and \ref{thm2}, and define the following two sets:
 \[\begin{split}
&C_n=\{\theta: \|\theta-\theta^*\|\leq M\sqrt n\sigma^*\epsilon_n, (1-\epsilon_n)/(1+\epsilon_n)<\sigma^2/\sigma^{*2}<(1+\epsilon_n)/(1-\epsilon_n)\}^c\backslash B_n,\\
&B_n=\{\theta: \mbox{Among all }\{\theta_i-\theta_{i-1}\}_{i=2}^n , \mbox{ there are at least }\delta
 \mbox{ of them are greater than }\sigma\epsilon_n/n\}.
\end{split}\]
Using the same arguments, one can still establish exponential separation results (\ref{test1}) and (\ref{test2}).

To establish (\ref{md}), we notice that 
\[\begin{split}
&\frac{m(y)}{f^*(y)}=\frac{\int_{\sigma^2}\int_\theta
\sigma^{*n}\exp\{-\|\theta^*-\theta+\varepsilon\|^2/\sigma^{2}\}
d\theta d\sigma^2
}{\sigma^{n}\exp\{-\|\varepsilon\|^2/\sigma^{*2}\}}\\
\geq&\int_{\sigma^2}\int_\theta\exp\left\{-\frac{\|\theta^*-\theta\|^2}{\sigma^2}
-2\frac{(\theta^*-\theta)^T\varepsilon}{\sigma^2}+\frac{\|\varepsilon\|^2}{\sigma^{*2}}-\frac{\|\varepsilon\|^2}{\sigma^{2}}-n\log(\sigma/\sigma^*)\right\}
\pi(\theta,\sigma^2)d\theta d\sigma^2\\
\geq& \pi(\max\{|\theta_i-\theta_i^*|\}/\sigma\leq \log n/n, 0\leq\sigma^2-\sigma^{*2}\leq \sigma^{*2}\log n/n)\exp\{-c_3'\log n \}
\end{split}\]
for some constant $c_3'$ and 
\[\begin{split}
&\pi(\max\{|\theta_i-\theta_i^*|\}/\sigma\leq \log n/n, 0\leq\sigma^2-\sigma^{*2}\leq\sigma^{*2} \log n/n)\\
\geq &\sum_{r}\pi(\max\{|\theta_{r(1)}-\theta_{r(1)}^*|, |\theta_{r(i)}-\theta_{r(i-1)}|\}/\sigma\leq  |G^*|\log n/n^2, 0\leq\sigma^2-\sigma^{*2}\leq\sigma^{*2} \log n/n|r)\pi(r).\\
\end{split}
\] 
This ensures (\ref{md}).

As for the prior probability of $B_n$,
if the scale parameter for the $t$ distribution is sufficiently small, i.e. $s=n^{-w}$ for some large $w$ and $\int_{\pm \epsilon_n/n^2}\pi_{t,s}(x)dx\geq 1-1/n^{1+u}$ 
for some sufficiently large $u$ where $\pi_{t,s}$ denotes the $t$ density
function with scale parameter $s$,
then for any ranking $r$,
\[
\pi(B_n|r)\geq 1-\pi(\max\{\theta_{r(i)}-\theta_{r(i-1)}\}\leq \sigma\epsilon_n/n^2 |r)
\geq 1-(1-1/n^{1+u})^n \approx n^{-u}.
\]
This hence implies that $-\log(\pi(B_n)) \geq u\log n $.

\vskip 0.2in
\bibliographystyle{plain}
\bibliography{ref}

\begin{thebibliography}{10}

\bibitem{andrews1974scale}
David~F Andrews and Colin~L Mallows.
\newblock Scale mixtures of normal distributions.
\newblock {\em Journal of the Royal Statistical Society. Series B
  (Methodological)}, pages 99--102, 1974.

\bibitem{Barron1998}
A.~Barron.
\newblock Information-theoretic characterization of bayes performance and the
  choice of priors in parametric and nonparametric problems.
\newblock In J.M. Bernardo, J.~Berger, A.~Dawid, and A.~Smith, editors, {\em
  Bayesian Statistics 6}, pages 27--52, 1998.

\bibitem{berger2014bayesian}
James~O Berger, Xiaojing Wang, and Lei Shen.
\newblock A bayesian approach to subgroup identification.
\newblock {\em Journal of biopharmaceutical statistics}, 24(1):110--129, 2014.

\bibitem{betancourt2017bayesian}
Brenda Betancourt, Abel Rodr{\'\i}guez, and Naomi Boyd.
\newblock Bayesian fused lasso regression for dynamic binary networks.
\newblock {\em Journal of Computational and Graphical Statistics},
  26(4):840--850, 2017.

\bibitem{BhattacharyaPPD2015}
Anirban Bhattacharya, Debdeep Pati, Natesh~S Pillai, and David~B Dunson.
\newblock Dirichlet-laplace priors for optimal shrinkage.
\newblock {\em Journal of the American Statistical Association},
  110:1479--1490, 2015.

\bibitem{CarvalhoPS2010}
C.M. Carvalho, N.G. Polson, and J.G. Scott.
\newblock The horseshoe estimator for sparse signals.
\newblock {\em Biometrika}, 97:465--480, 2010.

\bibitem{CastilloSHV2015}
I.~Castillo, J.~Schmidt-Hieber, and A.W. van~der Vaart.
\newblock Bayesian linear regression with sparse priors.
\newblock {\em \ANNALS}, pages 1986--2018, 2015.

\bibitem{Castillov2012}
Isma{\"e}l Castillo and Aad van~der Vaart.
\newblock Needles and straw in a haystack: Posterior concentration for possibly
  sparse sequences.
\newblock {\em The Annals of Statistics}, 40(4):2069--2101, 2012.

\bibitem{ChenC2008}
J.~Chen and Z.~Chen.
\newblock Extended bayesian information criteria for model selection with large
  model spaces.
\newblock {\em Biometrika}, 95:759--771, 2008.

\bibitem{ChenC2012}
J.~Chen and Z.~Chen.
\newblock Extended bic for small-$n$-large-$p$ sparse glm.
\newblock {\em Statistica Sinica}, 22:555--574, 2012.

\bibitem{FanL2001}
J.~Fan and R.~Li.
\newblock Variable selection via nonconcave penalized likelihood and its oracle
  properties.
\newblock {\em \JASA}, 96:1348--1360, 2001.

\bibitem{GhosalGV2000}
Subhashis Ghosal, Jayanta~K Ghosh, and Aad~W Van Der~Vaart.
\newblock Convergence rates of posterior distributions.
\newblock {\em Annals of Statistics}, 28(2):500--531, 2000.

\bibitem{GhosalV2007}
Subhashis Ghosal and Aad~W Van Der~Vaart.
\newblock Convergence rates of posterior distributions for noniid observations.
\newblock {\em Annals of Statistics}, 35(1):192--223, 2007.

\bibitem{HahnC2015}
P.R Hahn and C.M. Carvalho.
\newblock Decoupling shrinkage and selection in bayesian linear models: A
  posterior summary perspective.
\newblock {\em \JASA}, 110:435--448, 2015.

\bibitem{heller2005bayesian}
Katherine~A Heller and Zoubin Ghahramani.
\newblock Bayesian hierarchical clustering.
\newblock In {\em Proceedings of the 22nd international conference on Machine
  learning}, pages 297--304, 2005.

\bibitem{IshawaranR2005}
H.~Ishwaran and J.S. Rao.
\newblock Spike and slab variable selection: frequentist and bayesian
  strategies.
\newblock {\em Annals of Statistics}, pages 730--773, 2005.

\bibitem{Jiang2007}
W.~Jiang.
\newblock Bayesian variable selection for high dimensional generalized linear
  models: Convergence rate of the fitted densities.
\newblock {\em \ANNALS}, 35:1487--1511, 2007.

\bibitem{JohnsonR2012}
V.E. Johnson and D.~Rossel.
\newblock Bayesian model selection in high-dimensional settings.
\newblock {\em \JASA}, 107:649--660, 2012.

\bibitem{johnstone2010high}
Iain~M Johnstone.
\newblock High dimensional bernstein-von mises: simple examples.
\newblock {\em Institute of Mathematical Statistics collections}, 6:87, 2010.

\bibitem{ke2015homogeneity}
Zheng~Tracy Ke, Jianqing Fan, and Yichao Wu.
\newblock Homogeneity pursuit.
\newblock {\em Journal of the American Statistical Association},
  110(509):175--194, 2015.

\bibitem{KeFW2015}
Zheng~Tracy Ke, Jianqing Fan, and Yichao Wu.
\newblock Homogeneity pursuit.
\newblock {\em Journal of the American Statistical Association}, 110:175--194,
  2015.

\bibitem{kleijn2006misspecification}
Bas~JK Kleijn, Aad~W van~der Vaart, et~al.
\newblock Misspecification in infinite-dimensional bayesian statistics.
\newblock {\em The Annals of Statistics}, 34(2):837--877, 2006.

\bibitem{KleijnV2006}
B.J.K. Kleijn and A.W. van~der Vaart.
\newblock Misspecification in infinite-dimensional bayesian statistics.
\newblock {\em \ANNALS}, 34:837--877, 2006.

\bibitem{kyung2010penalized}
Minjung Kyung, Jeff Gill, Malay Ghosh, and George Casella.
\newblock Penalized regression, standard errors, and bayesian lassos.
\newblock {\em Bayesian Analysis}, 5(2):369--411, 2010.

\bibitem{LaurentM2000}
B{\'e}atrice Laurent and Pascal Massart.
\newblock Adaptive estimation of a quadratic functional by model selection.
\newblock {\em Annals of Statistics}, pages 1302--1338, 2000.

\bibitem{li2018spatial}
Furong Li and Huiyan Sang.
\newblock Spatial homogeneity pursuit of regression coefficients for large
  datasets.
\newblock {\em Journal of the American Statistical Association}, 2018.

\bibitem{LiP2017}
H.~Li and D.~Pati.
\newblock Variable selection using shrinkage priors.
\newblock {\em Computational Statistics \& Data Analysis}, 107:107–119, 2017.

\bibitem{LiangSY2013}
F.~Liang, Q.~Song, and K.~Yu.
\newblock Bayesian subset modeling for high dimensional generalized linear
  models.
\newblock {\em \JASA}, 108:589--606, 2013.

\bibitem{liu2010efficient}
Jun Liu, Lei Yuan, and Jieping Ye.
\newblock An efficient algorithm for a class of fused lasso problems.
\newblock In {\em Proceedings of the 16th ACM SIGKDD international conference
  on Knowledge discovery and data mining}, pages 323--332, 2010.

\bibitem{ma2017concave}
Shujie Ma and Jian Huang.
\newblock A concave pairwise fusion approach to subgroup analysis.
\newblock {\em Journal of the American Statistical Association},
  112(517):410--423, 2017.

\bibitem{Mozeika2018}
Alexander Mozeika and Anthony Coolen.
\newblock Mean-field theory of bayesian clustering.
\newblock {\em arXiv preprint arXiv:1709.01632}, 2018.

\bibitem{NarisettyH2014}
N.N. Narisetty and X.~He.
\newblock Bayesian variable selection with shrinking and diffusing priors.
\newblock {\em The Annals of Statistics}, 42(2):789--817, 2014.

\bibitem{neal2000markov}
Radford~M Neal.
\newblock Markov chain sampling methods for dirichlet process mixture models.
\newblock {\em Journal of computational and graphical statistics},
  9(2):249--265, 2000.

\bibitem{ParkC2008}
T.~Park and G.~Casella.
\newblock The bayesian lasso.
\newblock {\em \JASA}, 103:681--686, 2008.

\bibitem{rinaldo2009properties}
Alessandro Rinaldo et~al.
\newblock Properties and refinements of the fused lasso.
\newblock {\em The Annals of Statistics}, 37(5B):2922--2952, 2009.

\bibitem{robbins1985empirical}
Herbert Robbins.
\newblock An empirical bayes approach to statistics.
\newblock In {\em Herbert Robbins Selected Papers}, pages 41--47. Springer,
  1985.

\bibitem{royston1982algorithm}
JP~Royston.
\newblock Algorithm as 177: Expected normal order statistics (exact and
  approximate).
\newblock {\em Journal of the royal statistical society. Series C (Applied
  statistics)}, 31(2):161--165, 1982.

\bibitem{ScottB2010}
J.G. Scott and J.O. Berger.
\newblock Bayes and empirical-bayes multiplicity adjustment in the
  variable-selection problem.
\newblock {\em \ANNALS}, pages 2587--2619, 2010.

\bibitem{ShenH2012}
Xiaotong Shen and Hsin-Cheng Huang.
\newblock Grouping pursuit through a regularization solution surface.
\newblock {\em Journal of the American Statistical Association}, 105:727--739,
  2012.

\bibitem{shimamura2018bayesian}
Kaito Shimamura, Masao Ueki, Shuichi Kawano, and Sadanori Konishi.
\newblock Bayesian generalized fused lasso modeling via neg distribution.
\newblock {\em Communications in Statistics-Theory and Methods}, pages 1--23,
  2018.

\bibitem{SongL2014}
Q.~Song and F.~Liang.
\newblock A split-and-merge bayesian variable selection approach for ultra-high
  dimensional regression.
\newblock {\em \JRSSB}, in press, 2014.

\bibitem{SongL2017}
Qifan Song and Faming Liang.
\newblock Nearly optimal bayesian shrinkage for high dimensional regression.
\newblock {\em arXiv preprint arXiv:1712.08964}, 2017.

\bibitem{TangXGG2016}
X.~Tang, X.~Xu, M.~Ghosh, and P.~Ghosh.
\newblock Bayesian variable selection and estimation based on global-local
  shrinkage priors.
\newblock {\em arXiv:1605.07981}, 2016.

\bibitem{Tibshirani1996}
R.~Tibshirani.
\newblock Regression shrinkage and selection via the lasso.
\newblock {\em \JRSSB}, 58:267--288, 1996.

\bibitem{tibshirani2005sparsity}
Robert Tibshirani, Michael Saunders, Saharon Rosset, Ji~Zhu, and Keith Knight.
\newblock Sparsity and smoothness via the fused lasso.
\newblock {\em Journal of the Royal Statistical Society: Series B (Statistical
  Methodology)}, 67(1):91--108, 2005.

\bibitem{tibshirani2007spatial}
Robert Tibshirani and Pei Wang.
\newblock Spatial smoothing and hot spot detection for cgh data using the fused
  lasso.
\newblock {\em Biostatistics}, 9(1):18--29, 2007.

\bibitem{van2011statistics}
Sara van~der Geer and Peter B{\"u}hlmann.
\newblock {\em Statistics for High-Dimensional Data: Methods, Theory and
  Applications}.
\newblock Spring Series in Statistics, Springer, 2011.

\bibitem{VanSV2017}
S.L. van~der Pas, B.~Szabo, and Aad van~der Vaart.
\newblock Adaptive posterior contraction rates for the horseshoe.
\newblock {\em arXiv:1702.03698}, 2017.

\bibitem{wade2015bayesian}
Sara Wade and Zoubin Ghahramani.
\newblock Bayesian cluster analysis: Point estimation and credible balls.
\newblock {\em Bayesian Analysis}, 13:559--626, 2018.

\bibitem{XuSMQH2017}
Z.~Xu, D.F. Schmidt, E.~Makalic, G.~Qian, and J.L. Hopper.
\newblock Bayesian sparse global-local shrinkage regression for grouped
  variables.
\newblock {\em arXiv:1709.04333}, 2017.

\bibitem{YangWJ2015}
Yun Yang, Martin~J Wainwright, and Michael~I Jordan.
\newblock On the computational complexity of high-dimensional bayesian variable
  selection.
\newblock {\em \ANNALS}, in press, 2015.

\bibitem{Zhang2010}
C.-H. Zhang.
\newblock Nearly unbiased variable selection under minimax concave penalty.
\newblock {\em \ANNALS}, 38:894--942, 2010.

\bibitem{Zou2006}
H.~Zou.
\newblock The adaptive lasso and its oracle properties.
\newblock {\em \JASA}, 101:1418--1429, 2006.

\bibitem{ZubkovS2013}
AM~Zubkov and AA~Serov.
\newblock A complete proof of universal inequalities for the distribution
  function of the binomial law.
\newblock {\em Theory of Probability \& Its Applications}, 57:539--544, 2013.

\end{thebibliography}

\end{document}